\documentclass[11pt]{article}
\usepackage{geometry}
\usepackage{graphicx}
\usepackage{amssymb}
\usepackage{epstopdf}
\usepackage{caption}
\usepackage{hyperref}
\usepackage{multicol}
\usepackage{natbib}
\usepackage{booktabs}
\usepackage{amsthm}
\usepackage[english]{babel}
\usepackage{amssymb}
\usepackage{float}
\usepackage{bbm}
\usepackage{empheq}
\usepackage{amsmath}
\usepackage{lipsum}
\usepackage{docmute}
\usepackage{multirow}
\usepackage{setspace} 
\usepackage{hyperref}
\usepackage{mathtools}
\newtheorem{theorem}{Theorem}
\newtheorem{remark}[theorem]{Remark}
\newtheorem{lemma}[theorem]{Lemma}
\newtheorem{condition}[theorem]{Condition}

\setcitestyle{authoryear,open={(},close={)},citesep={;}}

\providecommand{\keywords}[1]
{
  \small	
  \textbf{\textit{Keywords---}} #1
}

\title{Inference for Delay Differential Equations Using Manifold-Constrained Gaussian Processes}

\author{Yuxuan Zhao and Samuel W.K. Wong\footnote{E-mail for correspondence: samuel.wong@uwaterloo.ca} \\
Department of Statistics and Actuarial Science, University of Waterloo, \\Waterloo, N2L 3G1, Canada.}

\begin{document}
\setstretch{1.25}

\maketitle

\begin{abstract}   
Dynamic systems described by differential equations often involve feedback among system components. When there are time delays for components to sense and respond to feedback, delay differential equation (DDE) models are commonly used. This paper considers the problem of inferring unknown system parameters, including the time delays, from noisy and sparse experimental data observed from the system. We propose an extension of manifold-constrained Gaussian processes to conduct parameter inference for DDEs, whereas the time delay parameters have posed a challenge for existing methods that bypass numerical solvers. Our method uses a Bayesian framework to impose a Gaussian process model over the system trajectory, conditioned on the manifold constraint that satisfies the DDEs. For efficient computation, a linear interpolation scheme is developed to approximate the values of the time-delayed system outputs, along with corresponding theoretical error bounds on the approximated derivatives. Two simulation examples, based on Hutchinson's equation and the \textit{lac} operon system, together with a real-world application using Ontario COVID-19 data, are used to illustrate the efficacy of our method. \\
\end{abstract}


\section{Introduction}

Delay differential equations (DDEs) are commonly used to model dynamic processes with time delays in ecology \citep{Hutchinson1948CIRCULARCS}, biology \citep{yildirim2003feedback}, and epidemiology \citep{ma2004global}. In this setting, components of the system require time to respond to feedback, thereby introducing time delay parameters into the model. For example, when modeling disease transmission mechanisms, the system of DDEs could incorporate constant time delays to accommodate the incubation period of the infectious disease \citep{ma2004global}. In models for gene regulation, the time required for transcription and translation activities could be encoded via time delay parameters
\citep{yildirim2003feedback}.

Our focus lies in models described by a set of DDEs,
\begin{equation}
dx_i(t)/dt = f_i (\boldsymbol{x}(t),\boldsymbol{x}(t-\boldsymbol{\tau}),\boldsymbol{\theta}, t), \quad t\in [0,T], i\in \{1,\cdots,m\},    
\label{eq:dde}
\end{equation}%
where $m$ is the number of system components and the vector $\boldsymbol{x}(t) = (x_1(t),x_2(t),\cdots, x_m(t))$ denotes the system output at time $t$. Here, we let $\boldsymbol{\tau} = (\tau_1,\tau_2,\cdots, \tau_m)$ denote the vector of time delay parameters for the system components (with $\tau_i = 0$ if the $i$-th component does not involve a time delay) and $\boldsymbol{\theta}$ denote the $l$-dimensional vector of other unknown model parameters. The term $\boldsymbol{x}(t-\boldsymbol{\tau}) = (x_1(t-\tau_{1}),x_2(t-\tau_{2}),\cdots, x_m(t-\tau_{m}))$ refers to the time-delayed system output at time $t-\boldsymbol{\tau}$, or \textit{historical output} for short. In many applications, the $m$ DDEs share the same vector of time delay parameters \citep{glass2021nonlinear}, which we write as $\boldsymbol{\tau}$ for simplicity of notation throughout this paper; if there are equation-specific time delays, we may instead write $\boldsymbol{\tau}_{i} = (\tau_{i,1},\tau_{i,2}, \ldots, \tau_{i,m})$ to denote the vector of time delay parameters in the $i$-th equation. 

In contrast to ordinary differential equations (ODEs), DDEs require information about the history of the system dynamics: we let $\mathcal{H}_{\tau_i} = \{x_i(t), t\in [-\tau_i,0]\}$ denote the values of the $i$-th system component prior to $t= 0$. The widely-used history function $\mathcal{H}_{\tau_i} = \{x_i(t) = x_i(0), t\in [-\tau_i,0]\}$ is assumed throughout this paper, which sets $x_i(t-\tau_i)= x_i(0)$ for all $t \le \tau_i$ \citep{yildirim2003feedback,ma2004global,bihorel2011scarabee,wang2012estimating,wang2022adaptive}.
The functions $f_i(\cdot):\mathbb{R}^m\times \mathbb{R}^m\times \mathbb{R}^l\times \mathbb{R}\to \mathbb{R}$ are treated as known from the scientific context. Due to experimental limitations, data observed from the system are usually noisy and collected at a discrete set of time points $\boldsymbol{\gamma}$. Letting $\boldsymbol{y}(\boldsymbol{\gamma})$ denote the observations, we assume $\boldsymbol{y}(\boldsymbol{\gamma}) = \boldsymbol{x}(\boldsymbol{\gamma}) +\boldsymbol{\epsilon}(\boldsymbol{\gamma})$, where $\boldsymbol{\epsilon}(\boldsymbol{\gamma})$ is Gaussian noise. The goal is to infer the unknown parameters in $\boldsymbol{\theta}$ and $\boldsymbol{\tau}$ given the observations $\boldsymbol{y}(\boldsymbol{\gamma})$. 

Inferring parameters in differential equations (DEs) can be a difficult task, as nonlinear systems tend to lack analytic solutions. Thus, the first and oldest general category of methods for the parameter inference problem uses (deterministic) numerical solvers. \citet{horbelt2002parameter} minimized the weighted sum of squares between the observed data and the model trajectories with respect to the unknown parameters; this nonlinear least squares approach does not require distributional assumptions. A likelihood-based approach, which estimates the parameters by maximizing the likelihood function of the observed data given the numerical solution, can also be used \citep{bihorel2011scarabee}. 
A corresponding Bayesian approach imposes priors on the parameters and can use Markov chain Monte Carlo (MCMC) samplers for inference \citep{boersch2017debinfer}.
While these numerical solver-based methods are better known in the context of ODEs, they are also generally applicable to DDE parameter inference with the help of a numerical solver that supports DDEs \citep{bellen2013numerical}. Even in the ODE setting, the repeated use of numerical solvers required for parameter estimation incurs a high computational cost.  
Numerically solving DDE systems further requires an appropriate interpolation scheme for the historical output; moreover, the historical output can reduce the stability of the system and introduce additional oscillations, so that the solver requires smaller discretization time steps.
Even a system that is stable when modeled with ODEs can become unstable with the addition of a delay \citep{erneux2009applied}. Therefore, numerical solvers for DDEs tend to be even more expensive than for ODEs.

To bypass numerical solvers, the class of collocation methods approximates the DE solution via basis function expansion \citep{varah1982spline,ramsay2007parameter,wang2012estimating,wang2022adaptive}. Introducing this idea for ODEs, \citet{varah1982spline} used splines to first smooth the noisy data and then in a second step applied least squares on the difference between the spline derivative and the ODEs; doing so requires a relatively accurate estimate of $\boldsymbol{x}(t)$ in the first step. \citet{ellner1997inferring} applied similar strategies for the case of DDEs. 
\citet{ramsay2007parameter} pioneered a generalized profiling procedure for ODEs, wherein the model parameters and spline coefficients are optimized together using a penalized likelihood, to achieve a balance between data fitting and fidelity to the ODE. 
Extending this approach to a semiparametric method for inference of DDEs, \citet{wang2012estimating} used cross-validation to choose the 
smoothing parameter in the penalized likelihood and noted that a denser set of knots may be needed to handle sharp changes and oscillations in the DDE solution. 
Subsequently, \citet{wang2022adaptive} proposed a more comprehensive Bayesian approach to collocation for ODE and DDE models, and developed an annealed sequential Monte Carlo (SMC) algorithm to draw posterior samples. By treating the smoothing parameter as part of the posterior distribution, the method of \citet{wang2022adaptive} avoids the potential drawbacks of over-fitting and expensive computation associated with cross-validation, and might be regarded as the state-of-the-art collocation-based inference method for DDEs.

Within a Bayesian framework, Gaussian processes (GPs) can serve as an alternative to collocation methods. For ODE inference, GP-based methods begin by placing a GP prior with hyper-parameters $\phi$ on $\boldsymbol{x}(t)$, which gives a closed form of the conditional distribution of $\boldsymbol{x}'(t)$ given $\boldsymbol{x}(t)$ to help bypass numerical solvers \citep{calderhead2008accelerating,dondelinger2013ode,wenk2019fast,yang2021inference}.  The earliest GP-based method used gradient matching together with heuristics to combine the GP and ODE specifications \citep{calderhead2008accelerating}; an application to DDEs was also considered, but in practice would require numerical solvers for the delay parameters. Subsequent refinements to GP gradient matching \citep[e.g.,][]{dondelinger2013ode,wenk2019fast} nonetheless did not resolve the incompatible specification of $\boldsymbol{x}'(t)$ between the GP (via $p(\boldsymbol{x}'(t)|\boldsymbol{x}(t),\phi)$) and the ODEs (via the functions $f_i$).
To address this incompatibility, \citet{yang2021inference} proposed the  manifold-constrained Gaussian process inference (MAGI) method for ODEs. In contrast to gradient matching, MAGI explicitly conditions $\boldsymbol{x}(t)$ on the manifold that satisfies the ODEs, thereby providing a coherent Bayesian framework for inference. MAGI demonstrated accurate parameter estimation and fast computation speed on sparse and noisy data; \citet{wong2023estimating} further showcased the promising performance of MAGI on a larger ODE model with 10 components and 16 unknown parameters. 
It is however challenging to incorporate historical outputs for DDEs into GP-based methods, since system outputs at a non-fixed set of time points are required (as $\boldsymbol{\tau}$ is random), which is computationally demanding due to the covariance structure of GPs. Historical outputs pose a similar challenge for probabilistic ODE solvers that employ fixed time steps \citep{tronarp2022fenrir,wu2024data}. To the best of our knowledge, the use of GPs to facilitate parameter inference for DDEs, without any numerical solvers, remains to be explored.

Therefore, as the main contribution of this paper, we develop methodology that extends the MAGI GP-based framework to DDEs, given its principled Bayesian construction and favorable performance on ODEs.
The key challenge is to incorporate historical outputs of DDEs into the GP manifold constraints, while achieving fast computation speed and estimation accuracy. We tackle this challenge by employing a linear interpolation scheme that only involves sparse matrix computations and simplifies the formulation of the manifold constraints. Theoretical error bounds for the resulting $\boldsymbol{x}'(t)$ along with numerical validations are provided to justify the approximation. Our proposed method, MAnifold-constrained Gaussian process inference for delay differential equations (MAGIDDE), is applied to simulated and real data examples. We obtain fast and accurate inference for the parameters and system trajectories, compared to other representative methods for DDE inference.

\section{Review of the MAGI method for ODE inference}
\label{sec:background}

In this section, we review the MAGI method \citep{yang2021inference} as applied to parameter inference for ODEs. In the ODE setting, Equation \eqref{eq:dde} reduces to $dx_i(t)/dt = f_i (\boldsymbol{x}(t),\boldsymbol{\theta}, t)$, $i = 1, \ldots, m$, or $d\boldsymbol{x}(t)/dt = \mathbf{f} (\boldsymbol{x}(t),\boldsymbol{\theta}, t)$ in vector form. MAGI imposes an independent GP prior on each component $x_i(t)$ such that \begin{equation}x_i(t)\sim \mathcal{GP}(\mu_i, \mathcal{K}_i), \quad t\in[0, T],
\label{eqn:x_prior}\end{equation} where $\mathcal{K}_i:\mathbb{R}\times \mathbb{R}\to \mathbb{R}$ is a positive definite covariance kernel and $\mu_i:\mathbb{R} \to \mathbb{R}$ is the mean function, usually taken to be $\mu_i(t) \equiv 0$. Let $\pi(\cdot)$ generically denote the prior on the model parameters $\boldsymbol{\theta}$.  

The data consist of the noisy observations $\boldsymbol{y}(\boldsymbol{\gamma}) = \left(\boldsymbol{y}_1(\boldsymbol{\gamma}_1),\cdots, \boldsymbol{y}_m(\boldsymbol{\gamma}_m)\right)$, where $\boldsymbol{\gamma} = (\boldsymbol{\gamma}_1,\cdots, \boldsymbol{\gamma}_m)$ denotes the collection of observation time points for each component, i.e., component $i$ is observed at the $N_i$ time points $\boldsymbol{\gamma}_i = (\gamma_{i,1},\cdots \gamma_{i,N_i})$. Assume that the observed data $\boldsymbol{y}_i(\boldsymbol{\gamma}_i)$ are subject to additive Gaussian noise, i.e., \begin{equation}\boldsymbol{y}_i(\boldsymbol{\gamma}_i) = \boldsymbol{x}_i(\boldsymbol{\gamma}_i)+\boldsymbol{\epsilon}_i(\boldsymbol{\gamma}_i), \quad \boldsymbol{\epsilon}_i(\boldsymbol{\gamma}_i)\sim N(\mathbf{0}, \sigma_i^2 \mathbf{I}_{N_i}),
    \label{eqn:obs_data}
\end{equation}
where $\mathbf{I}_{N_i}$ is an $N_i\times N_i$ dimensional identity matrix. (The notation $t$ is used to represent time generically, while $\boldsymbol{\gamma}$ denotes the observation time points.)  

With a suitably chosen $\mathcal{K}_i$, the conditional distribution of $\boldsymbol{x}'(t)$ given $\boldsymbol{x}(t)$ is also a GP with a fully specified mean function and a covariance kernel. To link this GP-specified distribution of $\boldsymbol{x}'(t)$ with the ODE model structure, define $W$ as a random variable measuring the uniform deviation between the stochastic process and ODE, i.e.,$$
    W = \sup_{t\in[0,T], i\in\{1,\cdots, m\}} \big|x'_i(t)-f_i(\boldsymbol{x}(t),  \boldsymbol{\theta},t)\big|.
$$
Setting $W=0$ thus constrains the GP to lie on the manifold that satisfies the ODEs. In practice, $W$ needs to be approximated by taking the uniform deviation over a finite set of $n$ discretization points $\boldsymbol{I} = \{t_1,\cdots, t_n\}$, so we define$$
    W_{\boldsymbol{I}} = \max_{t\in\boldsymbol{I}, i\in\{1,\cdots, m\}} \big|x'_i(t)-f_i\left(\boldsymbol{x}(t), \boldsymbol{\theta},t\right)\big|,
$$
with $\boldsymbol{\gamma}\subset \boldsymbol{I}\subset [0,T]$. The key idea is to then condition the GP on the manifold constraint $W=0$, as approximated by setting $W_{\boldsymbol{I}} = 0$. Following a Bayesian framework, the computable joint posterior of $\boldsymbol{\theta}$ and $\boldsymbol{x}(\boldsymbol{I})$ conditional on $W_{\boldsymbol{I}} = 0$ and the noisy measurements $\boldsymbol{y}(\boldsymbol{\gamma})$ is given by \begin{equation}
    \begin{aligned}
    &p\left(\boldsymbol{\theta}, \boldsymbol{x}(\boldsymbol{I})| W_{\boldsymbol{I}} = 0, \boldsymbol{y}(\boldsymbol{\gamma})\right)
    \propto p\left(\boldsymbol{\theta}, \boldsymbol{x}(\boldsymbol{I}), W_{\boldsymbol{I}} = 0,\boldsymbol{y}(\boldsymbol{\gamma})\right)\\
    =&\pi(\boldsymbol{\theta}) \times p\left(\boldsymbol{x}(\boldsymbol{I})\right) \times p\left(\boldsymbol{y}(\boldsymbol{\gamma})| \boldsymbol{x}(\boldsymbol{I})\right)\times p\left(\boldsymbol{x}'(\boldsymbol{I})=\mathbf{f}(\boldsymbol{x}(\boldsymbol{I}), \boldsymbol{\theta}, \boldsymbol{I}) |\boldsymbol{x}(\boldsymbol{I})\right).
\end{aligned} 
\label{eqn:posterior}
\end{equation}
A brief description of these four terms follows, with details of the closed forms for each term in Equation \eqref{eqn:posterior} provided in Appendix A. The prior density of the model parameters is $\pi(\boldsymbol{\theta})$. As any finite collection of random variables from the GP $\boldsymbol{x}(t)$ defined in Equation \eqref{eqn:x_prior} follows a multivariate normal distribution, $p\left(\boldsymbol{x}(\boldsymbol{I})\right)$ is multivariate normal. According to Equation \eqref{eqn:obs_data}, $p\left(\boldsymbol{y}(\boldsymbol{\gamma})| \boldsymbol{x}(\boldsymbol{I})\right)$ is the normal likelihood of the noisy observations. The last term evaluates the GP $\boldsymbol{x}'(\boldsymbol{I})$ at $\boldsymbol{x}'(\boldsymbol{I}) = \mathbf{f}(\boldsymbol{x}(\boldsymbol{I}), \boldsymbol{\theta}, \boldsymbol{I})$ to satisfy $W_{\boldsymbol{I}} = 0$, and hence is also multivariate normal.

Hamiltonian Monte Carlo \citep[HMC,][]{neal2011mcmc} is used to draw samples of $\boldsymbol{\theta}$ and $\boldsymbol{x}(\boldsymbol{I})$ from Equation \eqref{eqn:posterior}. Denser discretization sets $\boldsymbol{I}$ may provide more accurate inference as the manifold constraint is better approximated, at the cost of computation time. As a practical guideline for the choice of $\boldsymbol{I}$, \citet{wong2022magi} suggests starting with the smallest equally-spaced set $\boldsymbol{I}_0$ that contains the observation time points $\boldsymbol{\gamma}$. Then, we can construct $\boldsymbol{I}_j \supset \boldsymbol{I}_{j-1}$, $j \ge 1$ by inserting one equally-spaced point between each adjacent pair of points in $\boldsymbol{I}_j$, and stop when stable estimates are obtained (i.e., the credible intervals based on $\boldsymbol{I}_j$ and $\boldsymbol{I}_{j-1}$ largely overlap). \citet{yang2021inference} notes that as the cardinality $|\boldsymbol{I}|$ increases, the terms involving the GP prior in Equation \eqref{eqn:posterior}, namely $p\left(\boldsymbol{x}(\boldsymbol{I})\right)$ and $p\left(\boldsymbol{x}'(\boldsymbol{I}) = \mathbf{f}(\boldsymbol{x}(\boldsymbol{I}), \boldsymbol{\theta}, \boldsymbol{I}) |\boldsymbol{x}(\boldsymbol{I})\right)$, will become more dominant, while the likelihood remains unchanged, i.e., only the points in the observation set $\boldsymbol{\gamma}$ contribute to the likelihood term $p \left(\boldsymbol{y}(\boldsymbol{\gamma})| \boldsymbol{x}(\boldsymbol{I})\right)$.  To achieve a balance between the likelihood and GP prior, a tempering hyper-parameter $\beta = m|\boldsymbol{I}|/\sum_{i=1}^m|\boldsymbol{\gamma}_i|$ is introduced, i.e., the ratio between the total number of discretization points and the total number of observations. The GP-related terms are then tempered as $\left[p\left(\boldsymbol{x}(\boldsymbol{I})\right)p\left(\boldsymbol{x}'(\boldsymbol{I}) =\mathbf{f}(\boldsymbol{x}(\boldsymbol{I}), \boldsymbol{\theta}, \boldsymbol{I}) |\boldsymbol{x}(\boldsymbol{I})\right) \right]^{1/{\beta}}$. Finally, based on the HMC samples, the posterior means of $\boldsymbol{\theta}$ and $\boldsymbol{x}(\boldsymbol{I})$ are treated as the estimates of the parameters and system trajectories, respectively, and credible intervals can be obtained.

\section{Methodology}
\label{sec:method}
Our main contribution is the MAGIDDE method which provides fast and accurate inference for DDEs, built upon the MAGI inference framework. To fix ideas, we begin by considering a fully Bayesian construction that incorporates the time-delay parameters $\boldsymbol{\tau}$ and historical outputs $\boldsymbol{x}(t-\boldsymbol{\tau})$ into the posterior distribution. A key challenge that becomes apparent is that the historical outputs are not directly available and require a computationally intensive sampling step. 
To address this challenge, we consider two approximation schemes and derive their corresponding theoretical error bounds.  Our practical implementation of MAGIDDE uses a linear interpolation scheme for the historical outputs, which results in a computationally efficient method that maintains estimation accuracy.

\subsection{Fully Bayesian Scheme}

We first consider a fully Bayesian scheme that extends MAGI to the DDE setting. To incorporate the time-delay parameters, we place a general prior $\pi(\cdot)$ on $\boldsymbol{\tau}$. The full posterior distribution then consists of $\boldsymbol{\theta},\boldsymbol{\tau}, \boldsymbol{x}(\boldsymbol{I}),\boldsymbol{x}(\boldsymbol{I}-\boldsymbol{\tau})$ and can be written as 
\begin{equation}
\begin{aligned}
    &p\left(\boldsymbol{\theta},\boldsymbol{\tau}, \boldsymbol{x}(\boldsymbol{I}),\boldsymbol{x}(\boldsymbol{I}-\boldsymbol{\tau})| W_{\boldsymbol{I}} = 0, \boldsymbol{y}(\boldsymbol{\gamma})\right)
   \propto  p\left(\boldsymbol{\theta}, \boldsymbol{\tau}, \boldsymbol{x}(\boldsymbol{I}),\boldsymbol{x}(\boldsymbol{I}-\boldsymbol{\tau}), W_{\boldsymbol{I}} = 0, \boldsymbol{y}(\boldsymbol{\gamma})\right),\\
\end{aligned} 
\label{eqn:posterior_dde}
\end{equation}
where $\boldsymbol{x}(\boldsymbol{I}-\boldsymbol{\tau}) = (\boldsymbol{x}_1(\boldsymbol{I}-\tau_1),\cdots, \boldsymbol{x}_m(\boldsymbol{I}-\tau_m))$, and each $\boldsymbol{x}_i(\boldsymbol{I}-\tau_i) = \left(x_i(t_1-\tau_i),\cdots, x_i(t_n-\tau_i)\right)^\top$. 

Factorizing Equation \eqref{eqn:posterior_dde} yields
$$
    \begin{aligned}
        &p\left( \boldsymbol{\theta}, \boldsymbol{\tau},\boldsymbol{x}(\boldsymbol{I}),\boldsymbol{x}(\boldsymbol{I}-\boldsymbol{\tau}), W_{\boldsymbol{I}} = 0, \boldsymbol{y}(\boldsymbol{\gamma})\right)\\
       =&\underbrace{\pi(\boldsymbol{\theta})\times {\pi(\boldsymbol{\tau})}}_{(1)} 
\times  \underbrace{p\left(\boldsymbol{x}(\boldsymbol{I}-\boldsymbol{\tau}),\boldsymbol{x}(\boldsymbol{I})|\boldsymbol{\theta},\boldsymbol{\tau} \right)}_{(2)}
\times   \underbrace{p\left( \boldsymbol{y}(\boldsymbol{\gamma})|\boldsymbol{x}(\boldsymbol{I}-\boldsymbol{\tau}),\boldsymbol{x}(\boldsymbol{I}),\boldsymbol{\theta},\boldsymbol{\tau}\right)}_{(3)}\\
\times & \underbrace{p\left(W_{\boldsymbol{I}} = 0|\boldsymbol{y}(\boldsymbol{\gamma}), \boldsymbol{x}(\boldsymbol{I}-\boldsymbol{\tau}),\boldsymbol{x}(\boldsymbol{I}),\boldsymbol{\theta},\boldsymbol{\tau}\right)}_{(4)}.\\
    \end{aligned}
$$ 
Due to the prior independence between the GP and parameters in $\boldsymbol{\theta}$, the second term
simplifies as $p\left(\boldsymbol{x}(\boldsymbol{I}),\boldsymbol{x}(\boldsymbol{I}-\boldsymbol{\tau})|\boldsymbol{\tau} \right)$.
Likewise, the likelihood of the noisy observations at time points $\boldsymbol{\gamma}$ does not depend on $\boldsymbol{\theta}$, $\boldsymbol{\tau}$, and $\boldsymbol{x}(\boldsymbol{I}-\boldsymbol{\tau})$, 
so the third term simplifies to $p\left(\boldsymbol{y}(\boldsymbol{\gamma})|\boldsymbol{x}(\boldsymbol{I})\right)$. For the fourth term, after substituting the definition of $W_{\boldsymbol{I}}=0$, the resulting density of $\boldsymbol{x}'(\boldsymbol{I})$ 
 depends on $\boldsymbol{x}(\boldsymbol{I})$, $\boldsymbol{x}(\boldsymbol{I}-\boldsymbol{\tau}),$ and $\boldsymbol{\tau}$. 
Hence, we obtain
$$
    \begin{aligned}
        p\left(\boldsymbol{\theta}, \boldsymbol{\tau}, \boldsymbol{x}(\boldsymbol{I}),\boldsymbol{x}(\boldsymbol{I}-\boldsymbol{\tau}), W_{\boldsymbol{I}} = 0,\boldsymbol{y}(\boldsymbol{\gamma})\right)
        =&\underbrace{\pi(\boldsymbol{\theta})\times {\pi(\boldsymbol{\tau})}}_{(1)} 
    \times  \underbrace{p\left(\boldsymbol{x}(\boldsymbol{I}-\boldsymbol{\tau}),\boldsymbol{x}(\boldsymbol{I})|\boldsymbol{\tau} \right)}_{(2)}
        \times   \underbrace{p\left( \boldsymbol{y}(\boldsymbol{\gamma})|\boldsymbol{x}(\boldsymbol{I})\right)}_{(3)}\\
         \times  &\underbrace{p\left(\boldsymbol{x}'(\boldsymbol{I})=\mathbf{f}(\boldsymbol{x}(\boldsymbol{I}),\boldsymbol{x}(\boldsymbol{I}-\boldsymbol{\tau}), \boldsymbol{\theta}, \boldsymbol{I})|\boldsymbol{x}(\boldsymbol{I}-\boldsymbol{\tau}),\boldsymbol{x}(\boldsymbol{I}),\boldsymbol{\tau})\right)}_{(4)}.\\
    \end{aligned}
  $$

The first term is the prior density of the parameters $\boldsymbol{\theta}$ and $\boldsymbol{\tau}$. 
For the second term, the joint distribution of $\boldsymbol{x}_i(\boldsymbol{I}-\tau_i)$ and $\boldsymbol{x}_i(\boldsymbol{I})$ given $\tau_i$ for the $i-$th component is multivariate normal from the GP prior for $x_i(t)$,
i.e.,
\begin{equation} 
\begin{aligned}
    \boldsymbol{x}_i(\boldsymbol{I}-\tau_i),\boldsymbol{x}_i(\boldsymbol{I}) | \tau_i \sim N&\left(\begin{bsmallmatrix}
        \boldsymbol{\mu}_i(\boldsymbol{I}-\tau_i)\\
         \boldsymbol{\mu}_i(\boldsymbol{I})
    \end{bsmallmatrix}, \begin{bsmallmatrix}
\mathcal{K}_i(\boldsymbol{I}-\tau_i,\boldsymbol{I}-\tau_i) & 
  \mathcal{K}_i(\boldsymbol{I}-\tau_i,\boldsymbol{I}) \\
  \mathcal{K}_i(\boldsymbol{I},\boldsymbol{I}-\tau_i)&
\mathcal{K}_i(\boldsymbol{I},\boldsymbol{I})
\end{bsmallmatrix}\right),
    \label{eqn:conditional_dist}
\end{aligned}
\end{equation}
where $\boldsymbol{\mu}(\boldsymbol{I}-\tau_i)$,  $\mathcal{K}_i(\boldsymbol{I}-\tau_i,\boldsymbol{I}-\tau_i)$ 
and $\mathcal{K}_i(\boldsymbol{I},\boldsymbol{I}-\tau_i)$ respectively represent the mean function at $\boldsymbol{I}-\tau_i$, covariance matrix of $\boldsymbol{x}_i(\boldsymbol{I}-\tau_i)$ and cross-covariance matrix between $\boldsymbol{x}_i(\boldsymbol{I})$ and $\boldsymbol{x}_i(\boldsymbol{I}-\tau_i)$. 
The third term is the normal likelihood of the noisy observations. 
The fourth term evaluates the density of $\boldsymbol{x}'(\boldsymbol{I})$ at $\mathbf{f}(\boldsymbol{x}(\boldsymbol{I}),\boldsymbol{x}(\boldsymbol{I}-\boldsymbol{\tau}), \boldsymbol{\theta}, \boldsymbol{I})$, and by the property of GPs $\boldsymbol{x}'(\boldsymbol{I})|\boldsymbol{x}(\boldsymbol{I}-\boldsymbol{\tau}),\boldsymbol{x}(\boldsymbol{I}),\boldsymbol{\tau}$ has a multivariate normal distribution, provided the covariance kernel $\mathcal{K}$ is
associated with twice-differentiable curves, i.e.,
\begin{equation} 
\boldsymbol{x}_i'(\boldsymbol{I})|\boldsymbol{x}_i(\boldsymbol{I}-\tau_i),\boldsymbol{x}_i(\boldsymbol{I}),\tau_i \sim N\left(\boldsymbol{\mu}'_i(\boldsymbol{I})+\boldsymbol{m}_i\left(\begin{bsmallmatrix} \boldsymbol{x}_i(\boldsymbol{I}-\tau_i)\\
         \boldsymbol{x}_i(\boldsymbol{I})
    \end{bsmallmatrix}-\begin{bsmallmatrix}
        \boldsymbol{\mu}_i(\boldsymbol{I}-\tau_i)\\
         \boldsymbol{\mu}_i(\boldsymbol{I})
    \end{bsmallmatrix} \right), \boldsymbol{\zeta}_i\right), \label{eqn:xprime_fully}
\end{equation} 
where $\boldsymbol{m}_i  =\begin{bsmallmatrix}{ }^{\prime}\mathcal{K}_i(\boldsymbol{I}, \boldsymbol{I}-{\tau_i})& { }^{\prime}\mathcal{K}_i(\boldsymbol{I}, \boldsymbol{I})
\end{bsmallmatrix} \begin{bsmallmatrix}
\mathcal{K}_i(\boldsymbol{I}-\tau_i,\boldsymbol{I}-\tau_i) & 
  \mathcal{K}_i(\boldsymbol{I}-\tau_i,\boldsymbol{I}) \\
  \mathcal{K}_i(\boldsymbol{I},\boldsymbol{I}-\tau_i)&
\mathcal{K}_i(\boldsymbol{I},\boldsymbol{I})
\end{bsmallmatrix}^{-1}$, and $\boldsymbol{\zeta}_i  =\mathcal{K}_i^{\prime \prime}(\boldsymbol{I}, \boldsymbol{I})-
\boldsymbol{m}_i \begin{bsmallmatrix} \mathcal{K}_i'( \boldsymbol{I}-{\tau_i},\boldsymbol{I})\\ \mathcal{K}_i'(\boldsymbol{I}, \boldsymbol{I})\end{bsmallmatrix}$
with ${ }^{\prime} \mathcal{K}_i=\frac{\partial}{\partial s} \mathcal{K}_i(s, t), \mathcal{K}^{\prime}_i=\frac{\partial}{\partial t} \mathcal{K}_i(s, t) \text {, and } \mathcal{K}^{\prime \prime}_i=\frac{\partial^2}{\partial s \partial t} \mathcal{K}_i(s, t)$ for any two time points $s$ and $t$.

In practice, we suggest using the Matern class where the covariance of the $i$-th component between time points $s$ and $t$ is given by $\mathcal{K}_i(s, t)=\phi_{i,1} \frac{2^{1-\nu}}{\Gamma(\nu)}\left(\sqrt{2 \nu} \frac{d}{\phi_{i,2}}\right)^\nu B_\nu\left(\sqrt{2 \nu} \frac{d}{\phi_{i,2}}\right),$ where $ d=|s-t|$, $\Gamma$ is the Gamma function, $B_\nu$ is the modified Bessel function of the second kind, and $\nu$ is the degree of freedom. Reasonable values are $\nu = 2.01$ or $2.5$ to satisfy the requirement of twice-differentiable curves; $\nu=2.01$ is a good default choice that is suitable for rougher curves, while $\nu=2.5$ is adequate for smoother curves and has faster computation speed \citep{wong2022magi}. $\mathcal{K}_i(s,t)$ has two hyper-parameters, $\phi_{i,1}$ and $\phi_{i,2}$, that respectively control the overall variance and bandwidth of the $i$-th component. 

However, sampling the historical outputs from the conditional distribution in Equation \eqref{eqn:conditional_dist} (and subsequently $\boldsymbol{x}_i'(\boldsymbol{I})$ from Equation \eqref{eqn:xprime_fully}) is a computational bottleneck, as recalculation of cross-covariances and the conditional covariance matrix $\boldsymbol{\zeta}_i$ is required each time the value of $\tau_i$ is updated.
Thus in what follows, we describe two schemes for approximating $\boldsymbol{x}(\boldsymbol{I}-\boldsymbol{\tau})$: conditional expectation and linear interpolation. For both schemes, we construct the approximation and derive theoretical error bounds for the derivatives of the system outputs. 

\subsection{Approximation of Historical Outputs}
\label{sec:Approximation of historical outputs}

To reduce the computational cost of the fully Bayesian scheme, we consider deterministic approximations for $\boldsymbol{x}(\boldsymbol{I}-\boldsymbol{\tau})$ in terms of ${\boldsymbol{x}}(\boldsymbol{I})$ and $\boldsymbol{\tau}$. Let $\hat{\boldsymbol{x}}(\boldsymbol{I}-\boldsymbol{\tau})$ denote the approximated values of the historical outputs. By dropping the explicit dependence on $\boldsymbol{x}(\boldsymbol{I}-\boldsymbol{\tau})$ in Equation \eqref{eqn:posterior_dde} and treating it as known, the factorization of the joint posterior simplifies to 
$$\begin{aligned}
        p\left(\boldsymbol{\theta},\boldsymbol{\tau}, \boldsymbol{x}(\boldsymbol{I})| W_{\boldsymbol{I}} = 0, \boldsymbol{y}(\boldsymbol{\gamma})\right)
       \propto & p\left(\boldsymbol{\theta}, \boldsymbol{\tau}, \boldsymbol{x}(\boldsymbol{I}), W_{\boldsymbol{I}} = 0,\boldsymbol{y}(\boldsymbol{\gamma})\right)\\
       =&\pi(\boldsymbol{\theta})\times {\pi(\boldsymbol{\tau})}
       \times  p\left(\boldsymbol{x}(\boldsymbol{I})|\boldsymbol{\theta},\boldsymbol{\tau}\right)
       \times   p\left(\boldsymbol{y}(\boldsymbol{\gamma})|\boldsymbol{x}(\boldsymbol{I}),\boldsymbol{\theta},\boldsymbol{\tau}\right)\\
        \times & p\left(\boldsymbol{x}'(\boldsymbol{I}) =\mathbf{f}(\boldsymbol{x}(\boldsymbol{I}),\boldsymbol{x}(\boldsymbol{I}-\boldsymbol{\tau}), \boldsymbol{\theta}, \boldsymbol{I})|\boldsymbol{y}(\boldsymbol{\gamma}), \boldsymbol{x}(\boldsymbol{I}),\boldsymbol{\theta},\boldsymbol{\tau}\right).\\
    \end{aligned} 
    $$
Then, using the previous conditional independence properties and substituting the approximation $\hat{\boldsymbol{x}}(\boldsymbol{I}-\boldsymbol{\tau})$ for $\boldsymbol{x}(\boldsymbol{I}-\boldsymbol{\tau})$ yields
\begin{equation}\begin{aligned}
    p\left(\boldsymbol{\theta},\boldsymbol{\tau}, \boldsymbol{x}(\boldsymbol{I})| W_{\boldsymbol{I}} = 0, \boldsymbol{y}(\boldsymbol{\gamma})\right)
   \propto & \pi(\boldsymbol{\theta})\times {\pi(\boldsymbol{\tau})}
   \times  p\left(\boldsymbol{x}(\boldsymbol{I})\right)
   \times   p\left( \boldsymbol{y}(\boldsymbol{\gamma})| \boldsymbol{x}(\boldsymbol{I})\right)\\
   \times &  p\left(\boldsymbol{x}'(\boldsymbol{I}) =\mathbf{f}(\boldsymbol{x}(\boldsymbol{I}),\hat{\boldsymbol{x}}(\boldsymbol{I}-\boldsymbol{\tau}), \boldsymbol{\theta}, \boldsymbol{I})|\boldsymbol{x}(\boldsymbol{I}), \boldsymbol{\tau}\right).\\
\end{aligned} 
\label{eqn:posterior_approximate}\end{equation}

A first approximation scheme for $\boldsymbol{x}(\boldsymbol{I}-\boldsymbol{\tau})$ is via conditional expectation. Specifically, $E\left[\boldsymbol{x}_i(\boldsymbol{I}-\tau_i)|\boldsymbol{x}_i(\boldsymbol{I})\right]$ stands as a natural candidate to approximate  $\boldsymbol{x}_i(\boldsymbol{I}-\tau_i)$, since it is the best linear predictor and has a closed-form expression \citep{wang2020prediction}. Thus, we have the approximation
\begin{equation}\hat{\boldsymbol{x}}_i(\boldsymbol{I}-\tau_i)=E\left[\boldsymbol{x}_i(\boldsymbol{I}-\tau_i)|\boldsymbol{x}_i(\boldsymbol{I})\right]=\boldsymbol{\mu}_i(\boldsymbol{I}-\tau_i) + \mathcal{K}_i(\boldsymbol{I},\boldsymbol{I}-\tau_i)\mathcal{K}_i(\boldsymbol{I},\boldsymbol{I})^{-1}(\boldsymbol{x}_i(\boldsymbol{I})-\boldsymbol{\mu}_i(\boldsymbol{I})).
    \label{eqn:conditional_expectation_approximation}
\end{equation}

A second approximation scheme for $\boldsymbol{x}(\boldsymbol{I}-\boldsymbol{\tau})$ uses linear interpolation. Define the interpolated weight $w(t) = \sum_{j=1}^{n-1}\mathbbm{1}(t\in [t_j, t_{j+1}])\cdot \frac{t-t_j}{t_{j+1}-t_j}$, where $\mathbbm{1}(\cdot)$ is the indicator function. Then a linearly interpolated value of the $i$-th component at any $t \in \{\boldsymbol{I} - \tau_i \}$ 
is given by  
\begin{equation}
\hat{x}_i(t) = \sum_{j=1}^{n-1}\mathbbm{1}(t\in [t_j, t_{j+1}])\cdot\left[(1-w(t))\cdot x_i(t_j)+w(t)\cdot x_i(t_{j+1})\right]
\label{eqn:linear_interpolation_approximation}.
\end{equation}
Equation \eqref{eqn:linear_interpolation_approximation} can be written in matrix form. Define a $n\times n$ scalar matrix $\boldsymbol{S}_i$ with its $(j,q)$-th entry given by
\begin{empheq}[left={S_i(j,q)=\empheqlbrace}]{align*}
  & 1-w(t_j-\tau_i) \quad &&\mbox{for } j\in \{1,\cdots, n\}, q=k(t_j-\tau_i)\\
  & w(t_j-\tau_i) \quad &&\mbox{for }  j\in \{1,\cdots, n\}, q=k(t_j-\tau_i)+1\\
  & 0 \quad &&\text{otherwise}
\end{empheq}
where $k(t_j-\tau_i)$ denotes the integer $c\in\{1,\cdots, n\}$ satisfying $t_j-\tau_i\in[t_c, t_{c+1}]$. Then Equation \eqref{eqn:linear_interpolation_approximation} can be rewritten as $\hat{\boldsymbol{x}}_i(\boldsymbol{I}-\tau_i) = \boldsymbol{S}_i\cdot \boldsymbol{x}_i(\boldsymbol{I})$.

Note that both approximation schemes can be viewed as a deterministic linear transformation of $\boldsymbol{x}(\boldsymbol{I})$. Given the history function $\mathcal{H}_{\tau_i}$, the approximation schemes are only applied to the historical output $x_i(t-\tau_i)$ for $t\in\boldsymbol{I}$ where $t > \tau_i$, as  we have $x_i(t-\tau_i) = x_i(0)$ for $t\le \tau_i$. Conditional expectation eliminates the need to work with the full covariance matrix of $(\boldsymbol{x}_i(\boldsymbol{I}-\tau_i), \boldsymbol{x}_i(\boldsymbol{I}))^\top$ and draw samples of ${\boldsymbol{x}}_i(\boldsymbol{I}-\tau_i)$ 
from the larger joint posterior in Equation \eqref{eqn:posterior_dde}. However, it still requires computation of the cross-covariance $\mathcal{K}_i(\boldsymbol{I},\boldsymbol{I}-\tau_i)$ and dense matrix multiplication. The computation cost can be substantively reduced further using linear interpolation, as then only sparse matrix ($\boldsymbol{S}_i$) operations are needed.

We next investigate the approximation error of $\mathbf{f}(\boldsymbol{x}(\boldsymbol{I}),\hat{\boldsymbol{x}}(\boldsymbol{I}-\boldsymbol{\tau}), \boldsymbol{\theta}, \boldsymbol{I})$ under the two proposed schemes. It suffices to find the stochastic error bound for any $i\in\{1,\cdots,m\},$ of $\max_{t\in\boldsymbol{I}}\bigg|f_i(\boldsymbol{x}(t),\boldsymbol{x}(t-\boldsymbol{\tau}),\boldsymbol{\theta},t)- f_i(\boldsymbol{x}(t),\hat{\boldsymbol{x}}(t-\boldsymbol{\tau}),\boldsymbol{\theta},t)\bigg|.$ 
Under some mild differentiability conditions (listed in Appendix B), we have Theorems \ref{thrm:pf_conditional_expectation} and  \ref{thrm:pf_linear_interpolation}, with proofs provided in Appendices C and D. These results guarantee that the approximation $f_i(\boldsymbol{x}(t),\hat{\boldsymbol{x}}(t-\boldsymbol{\tau}),\boldsymbol{\theta},t) $ will be reasonable under either scheme, given that $|\boldsymbol{I}|$ is sufficiently large.

 \begin{theorem}
    \label{thrm:pf_conditional_expectation}
    Suppose $x_i(t)$ is a GP with mean function $\mu_i(t)$ and Matern covariance kernel $\mathcal{K}_i$ with $\nu \in \{2.01, 2.5\}$, for $i = 1, \cdots, m$, and that the $n$ points in the discretization set $\boldsymbol{I}$ are equally-spaced over $[0,T]$. If the values of the historical outputs are approximated by conditional expectation as described in Equation \eqref{eqn:conditional_expectation_approximation}, then for any $i\in\{1,\cdots, m\}$, $$ \max_{t\in\boldsymbol{I}}\bigg|f_i(\boldsymbol{x}(t),\boldsymbol{x}(t-\boldsymbol{\tau}),\boldsymbol{\theta},t)- f_i(\boldsymbol{x}(t),\hat{\boldsymbol{x}}(t-\boldsymbol{\tau}),\boldsymbol{\theta},t)\bigg|= O_P\left((\log(n))^{\frac{1}{2}}n^{-\nu}\right).$$
 \end{theorem}

\begin{theorem}
    \label{thrm:pf_linear_interpolation} Consider the same setup as in Theorem \ref{thrm:pf_conditional_expectation}. If the values of the historical outputs are approximated by linear interpolation as described in Equation \eqref{eqn:linear_interpolation_approximation}, then for any $i\in\{1,\cdots, m\}$,
   $$ \max_{t\in\boldsymbol{I}}\bigg|f_i(\boldsymbol{x}(t),\boldsymbol{x}(t-\boldsymbol{\tau}),\boldsymbol{\theta},t)- f_i(\boldsymbol{x}(t),\hat{\boldsymbol{x}}(t-\boldsymbol{\tau}),\boldsymbol{\theta},t)\bigg|= O_P\left((\log(n))^{\frac{1}{2}}n^{-1.5}\right).$$ 
\end{theorem}
\begin{remark}
Theorems \ref{thrm:pf_conditional_expectation} and \ref{thrm:pf_linear_interpolation} assume an equally-spaced discretization set for $\boldsymbol{I}$. While the MAGIDDE method does not require an equally-spaced $\boldsymbol{I}$, these results suggest that a good practical guideline is to choose an equally-spaced (or approximately equally-spaced) set that includes the observation times.
\end{remark}
\begin{remark}
$\nu$ controls the level of smoothness of the Matern covariance; larger values of $\nu$ provide smoother covariance functions. If the degree of freedom in the Matern covariance kernel is set as $\nu = 1.5$, the rate of convergence of $f_i(\boldsymbol{x}(t),\hat{\boldsymbol{x}}(t-\boldsymbol{\tau}),\boldsymbol{\theta},t)$ will be similar whether we approximate the historical outputs by linear interpolation or conditional expectation.
Noting that MAGIDDE chooses $\nu \in \{2.01, 2.5\}$, Theorems \ref{thrm:pf_conditional_expectation} and \ref{thrm:pf_linear_interpolation} indicate that the conditional expectation scheme has a faster convergence rate than linear interpolation. \label{remark:conv_rate}
\end{remark}
\begin{remark}
    If $\boldsymbol{\tau}$ in Equation \eqref{eq:dde} is replaced by an equation-specific vector of time-delay parameters $\boldsymbol{\tau}_i = (\tau_{i,1},\tau_{i,2}, \ldots, \tau_{i,m})$ for the $i$-th DDE, the results of Theorems \ref{thrm:pf_conditional_expectation} and \ref{thrm:pf_linear_interpolation} on $\max_{t\in\boldsymbol{I}}\bigg|f_i(\boldsymbol{x}(t),\boldsymbol{x}(t-\boldsymbol{\tau}_i),\boldsymbol{\theta},t)- f_i(\boldsymbol{x}(t),\hat{\boldsymbol{x}}(t-\boldsymbol{\tau}_i),\boldsymbol{\theta},t)\bigg|$ still hold under the same differentiability conditions, by simply substituting $\boldsymbol{\tau}$ with $\boldsymbol{\tau}_i$.
\end{remark}

Remark \ref{remark:conv_rate} suggests a trade-off between the faster theoretical convergence rate of conditional expectation and faster computation speed of linear interpolation. We numerically validate and compare the performance of the three schemes on a benchmark system in Appendix E.  The results (Supplementary Table 1, Supplementary Figures 1 and 2) indicate that linear interpolation has a significant speed advantage over the other schemes and yields comparable estimates of parameters and trajectories when $|\boldsymbol{I}|$ is large enough, hence is our preferred choice in practice.

\subsection{Practical Implementation}
\label{sec:implementation}
The practical steps to implement the MAGIDDE method with the linear approximation scheme are as follows.
We begin by initializing the required parameter and hyper-parameter values needed for HMC sampling. 
First, we fit a GP to the noisy observations $\boldsymbol{y}(\boldsymbol{\gamma}_i)$ for each component $i$. If the noise level $\sigma_i^2$ is unknown, we obtain values of $\phi_{i,1}$, $\phi_{i,2}$ and $\sigma_i^2$ by maximizing the marginal likelihood $p(\phi_{i,1}, \phi_{i,2}, \sigma_i^2|\boldsymbol{y}(\boldsymbol{\gamma}_i))$; if $\sigma_i^2$ is known, we maximize the marginal likelihood $p(\phi_{i,1}, \phi_{i,2}|\boldsymbol{y}(\boldsymbol{\gamma}_i))$ instead. The obtained values of the covariance hyper-parameters $\phi_{i,1}$ and $\phi_{i,2}$ are held fixed during subsequent HMC sampling, as in \citet{yang2021inference}. 
Second, we use the observations $\boldsymbol{y}(\boldsymbol{\gamma})$ as the starting value of $\boldsymbol{x}(\boldsymbol{\gamma})$ for HMC sampling, and $\boldsymbol{x}(\boldsymbol{I})$ for time points in $\boldsymbol{I}\backslash\boldsymbol{\gamma}$ is initialized by linearly interpolating the observations.
Third, we optimize the posterior in Equation \eqref{eqn:posterior_approximate} with respect to $\boldsymbol{\theta}$ and $\boldsymbol{\tau}$ (holding other quantities fixed at their initialized values) to provide their corresponding starting values for HMC sampling. 
Recall that under the approximation scheme, it is not necessary to explicitly initialize $\hat{\boldsymbol{x}}(\boldsymbol{I}-\boldsymbol{\tau})$, since it can be viewed as a function of $\boldsymbol{\tau}$ and  $\boldsymbol{x}(\boldsymbol{I})$.

We then proceed to jointly sample $\boldsymbol{x}(\boldsymbol{I})$, together with unknown parameters in $\boldsymbol{\theta}$, $\boldsymbol{\tau}$, and $\boldsymbol{\sigma^2}$ 
from their posterior distribution (Equation \eqref{eqn:posterior_approximate}) using HMC \citep{neal2011mcmc}. Recall that we fix the GP hyper-parameters and use the tempering scheme described in Section \ref{sec:background} to balance the contributions of the GP prior and likelihood during HMC sampling. 
Our implementation of HMC adjusts the step sizes of the leapfrog integrator automatically during the burn-in phase to ensure that the acceptance rate falls within the range of 60$\%$ to 90$\%$.  
After discarding the samples from the burn-in period, the posterior means of $\boldsymbol{\theta}$ and $\boldsymbol{\tau}$ are treated as the parameter estimates. We call the posterior mean of $\boldsymbol{x}(\boldsymbol{I})$ the \textit{inferred trajectory}, which represents our estimate of the system trajectory. The uncertainties of $\boldsymbol{\theta}$, $\boldsymbol{\tau}$, and $\boldsymbol{x}(\boldsymbol{I})$ are quantified via 95$\%$ pointwise credible intervals.

\section{Simulation Studies}
\label{sec:simulation}
We assess the performance of MAGIDDE and compare with other representative methods for DDE inference, using Hutchinson's equation \citep{may1976simple} as a benchmark system in Section \ref{sec:hutchinson_simulation}. To further demonstrate the capability of MAGIDDE to handle more complex models, we consider the $lac$ operon system of \citet{yildirim2003feedback} in Section \ref{sec:lac_operon_system}.
\subsection{Benchmark System}
\label{sec:hutchinson_simulation}

As a benchmark system, we focus on parameter inference for Hutchinson's equation to evaluate the estimation accuracy of different methods under varying observation sample sizes. Hutchinson's equation was proposed for modelling single-species population dynamics, and \citet{may1976simple} introduced the following DDE for the population size of blowflies \citep{nicholson1954compensatory},
$dP(t)/dt = rP(t)\left[1-P(t-\tau)/(1000\cdot K)\right],$
where $P(t)$ is the population size at time $t$, $r$ is the reproduction rate, $\tau$ is a time delay, and $1000\cdot K$ is the maximum population sustainable by the limited food supply. A noisy observation taken at time $t$ is assumed to follow a lognormal distribution with mean $P(t)$ and variance $\sigma^2$ \citep{wang2022adaptive}. Since blowfly counts are strictly positive, we consider a logarithmic transformation by defining $P(t) = \exp\{N(t)\}$. Then, the log-transformed DDE is \begin{equation}
dN(t)/dt = r\left[1-\exp(N(t-\tau))/(1000\cdot K)\right],
\label{eqn:dde_eqn_log}
\end{equation}
and we can equivalently estimate $r$, $K$ and $\tau$ in Equation \eqref{eqn:dde_eqn_log} based on the noisy observations $y(t) \sim N(N(t), \sigma^2)$ for $t \in \boldsymbol{\gamma}$.

Following \citet{wang2022adaptive}, we set the true values of the parameters $(\boldsymbol{\theta},\boldsymbol{\tau}) = (r, K, \tau) = (0.8, 2, 3)$, the initial condition on the log-scale as $N(0) = \log(3500)$, the noise level $\sigma = 0.1$, and the history function $\mathcal{H}_\tau = \{N(t) = N(0), t\in[-\tau, 0]\}$. We also follow the setup in \citet{wang2012estimating} and take the time interval of interest as $t \in [0,30]$. Simulations in previous studies focused on a relatively dense observation set; however, estimation performance under more challenging scenarios with sparser observations also deserves attention.
Therefore, we consider scenarios with  $|\boldsymbol{\gamma}|=$ 16, 31, 61, and 121 equally-spaced observations over [0,30]. To create the simulation data, the R package $\mathit{deSolve}$ \citep{soetaert2010solving} is used to numerically solve the DDE trajectories as defined by Equation \eqref{eqn:dde_eqn_log}. 

We use two metrics to compare the performance of different methods. 
First, we calculate the root mean square error (RMSE) of the estimated parameters relative to the truth. 
Second, following \citet{yang2021inference}, we also calculate the \textit{trajectory RMSE} to assess the accuracy of the system trajectory implied by the parameter estimates. 
The trajectory RMSE is computed by the following steps: first, we use a numerical solver to construct the true trajectory based on the true parameter values and history function; then, we use the numerical solver to reconstruct the trajectory implied by the estimates of the parameters and history function. 
Last, we calculate the RMSE between the true trajectory and the reconstructed trajectory at the observation time points on the original scale of the measurements.

We compare MAGIDDE with two other representative methods for DDE inference: the \textit{deBInfer} R package and the semiparametric Bayesian collocation method \citep[SMCDE,][]{wang2022adaptive}. 
Briefly, \textit{deBInfer} uses MCMC to conduct parameter inference with the help of a numerical solver. SMCDE represents the solution trajectory via a basis function expansion and uses SMC to sample the parameters and associated trajectories. For each of these three methods, we choose similar flat priors for the model parameters to ensure fair comparison. We assess the performance for each method using 300 simulated datasets from the system.  We briefly describe how these three methods are run to conduct parameter inference for each simulated dataset in Appendix F.1.

For each of the three methods, we compute the RMSEs of the parameters and reconstructed trajectories, as summarized in 
Table \ref{tab:para_model_comparison} for the 300 simulated datasets. 
Supplementary Figures 3 and 4 in Appendix F.2 display the boxplots of the parameter estimates and trajectory RMSEs for each individual dataset.
Both MAGIDDE and SMCDE exhibit a similar pattern, that increasing the number of observations leads to more accurate parameter estimates; however, deBInfer does not follow this pattern as its random-walk MCMC sampler can struggle to efficiently explore the parameter space, regardless of the number of observations (as evidenced by a number of outliers in the boxplots as shown in Supplementary Figures 3 and 4).
MAGIDDE consistently outperforms the other two methods in terms of recovering the trajectory and parameters with lower error. Sparse observations do not inherently pose a problem for MAGIDDE, since the accuracy of the linear approximation scheme depends on the number of discretization points, and not the number of observations. To visualize the quality of the inferred trajectories, Figure \ref{fig:model_performance_magi} shows that MAGIDDE well-recovers the true underlying trajectory in the sparse 16 observation scenario; 
the 95$\%$ pointwise credible interval becomes narrower for scenarios with denser observations, e.g., Supplementary Figure 5 in Appendix G shows the corresponding plot for 61 observations.
MAGIDDE is also the fastest of the three methods, requiring about 26 seconds for the scenarios with $|\boldsymbol{I}| = 61$ (i.e., 16, 31, and 61 observations) and just over 1 minute for the scenario with 121 observations, using our R package implementation. While SMCDE is reasonably fast, deBInfer is an order of magnitude slower due to its dependence on the numerical DDE solver.

\begin{table}[!htbp]
\caption{Average parameter estimates (with RMSEs in parentheses) for the log-transformed Hutchinson's equation and average trajectory RMSEs over the 300 simulated datasets. The last column gives the average runtime (in minutes, on a single CPU core).}
\resizebox{\textwidth}{!}{
\centering
\begin{tabular}{@{}rlrrrrrrr@{}}
\toprule
\multicolumn{1}{c}{$|\boldsymbol{\gamma}|$}       & Method                  & \multicolumn{1}{c}{$r$} & \multicolumn{1}{c}{$K$} & \multicolumn{1}{c}{$\tau$} & \multicolumn{1}{c}{$N(0)$} & \multicolumn{1}{c}{$\sigma$} & \multicolumn{1}{c}{Traj.}& \multicolumn{1}{c}{Runtime} \\ \midrule
\multirow{1}{*}{16}  & MAGIDDE                    & 0.80($\mathbf{<0.01}$)              & 2.00($\mathbf{0.07}$)              & 3.00($\mathbf{0.01}$)                 & 8.16($\mathbf{0.04}$)                 & 0.11($\mathbf{0.02}$)                   & $\mathbf{96.75}$        &    $\mathbf{0.43}$      \\
                           & deBInfer                  & 0.67(0.33)              & 1.66(0.81)              & 2.56(1.02)                 & 8.14(0.09)                 & 0.12(0.05)                               & 628.10         &  14.4      \\
                           &SMCDE & 0.70(0.10)              & 2.01(1.08)              & 3.14(0.15)                 & 8.14(0.10)                 & 0.16(0.06)                   & 211.08  &     1.55        \\ \\
\multirow{1}{*}{31}  & MAGIDDE                    & 0.80($\mathbf{<0.01}$)              & 2.00($\mathbf{0.05}$)              & 3.00($\mathbf{0.01}$)                 & 8.16($\mathbf{0.03}$)                 & 0.10($\mathbf{0.01}$)                   & $\mathbf{74.69}$      &    $\mathbf{0.43}$       \\
                           & deBInfer                 & 0.66(0.33)              & 1.67(0.80)              & 2.57(1.02)                 & 8.11(0.14)                 & 0.12(0.04)                                & 589.05    &  19.1        \\
                           & SMCDE & 0.76(0.05)              & 2.01(0.14)              & 2.96(0.07)                 & 8.15(0.10)                 & 0.16(0.06)                   & 205.99 &       1.99       \\ \\
\multirow{1}{*}{61}  & MAGIDDE                    & 0.80($\mathbf{<0.01}$)              & 2.00($\mathbf{0.04}$)              & 3.00($\mathbf{0.01}$)                 & 8.16($\mathbf{0.02}$)                 & 0.10($\mathbf{0.01}$)                   & $\mathbf{55.14}$      &    $\mathbf{0.44}$         \\
                           & deBInfer               & 0.69(0.29)              & 1.73(0.71)              & 2.64(0.92)                 & 8.12(0.16)                 & 0.12(0.04)                    & 521.23   & 27.0         \\
                           & SMCDE & 0.76(0.04)              & 2.00(0.06)              & 2.94(0.07)                 & 8.15(0.09)                 & 0.16(0.06)                   & 200.47 &        2.76      \\ \\
\multirow{1}{*}{121} & MAGIDDE                    & 0.80($\mathbf{<0.01}$)              & 2.00($\mathbf{0.02}$)              & 3.00($\mathbf{0.01}$)                 & 8.16($\mathbf{0.02}$)                 & 0.10($\mathbf{0.01}$)                   & $\mathbf{32.59}$      &        $\mathbf{1.16}$     \\
                           & deBInfer                  & 0.65(0.34)              & 1.64(0.83)              & 2.50(1.14)                 & 8.09(0.20)                 & 0.12(0.04)               &  626.89    &    38.9       \\
                           & SMCDE & 0.77(0.04)              & 2.00(0.03)              & 2.94(0.07)                 & 8.14(0.08)                 & 0.16(0.06)                   &177.20 &     3.91        \\ \bottomrule
\end{tabular}}
\label{tab:para_model_comparison}
\end{table}

\begin{figure}[!htbp]
    \centering
    \includegraphics[width = \textwidth]{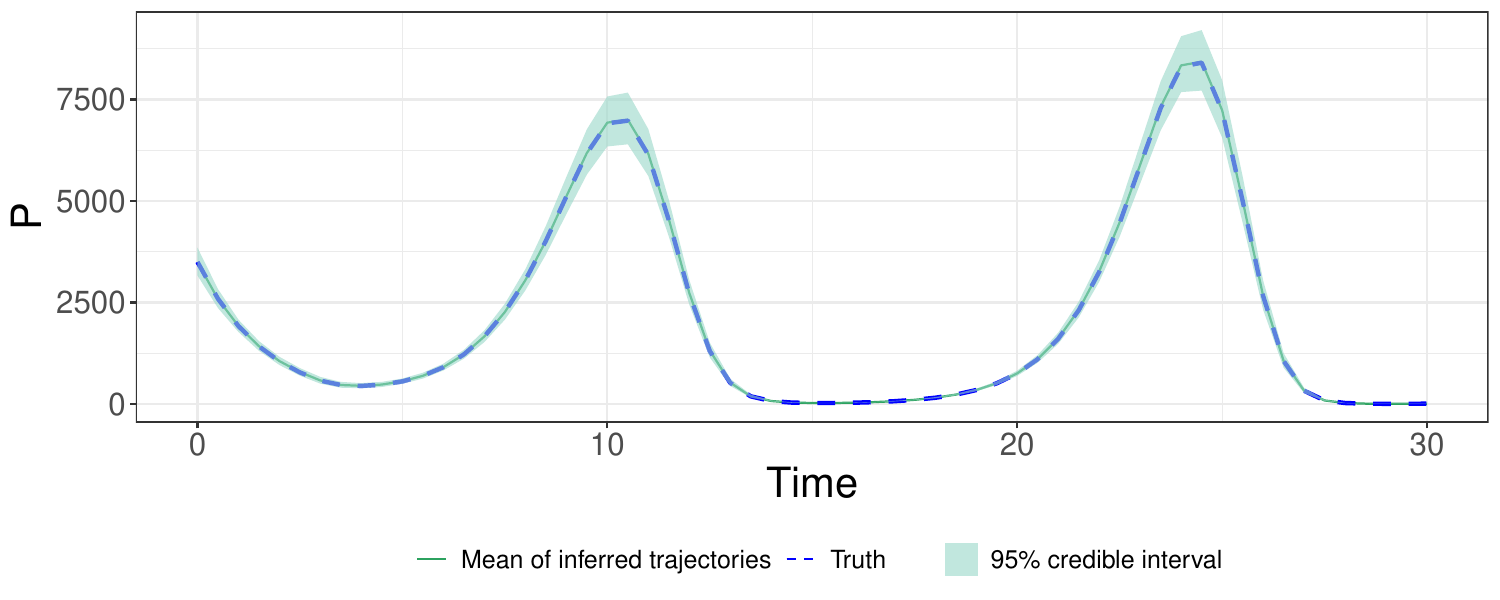}
    \caption{Inferred trajectory of Hutchinson's equation from 16 observations using MAGIDDE. The green solid line represents the mean inferred trajectory over 300 simulated datasets, and the blue dashed line is the truth. The green shaded area is the 95$\%$ pointwise credible interval, constructed by taking the average 0.025 and 0.975 quantiles of the inferred trajectories across the simulated datasets. }
    \label{fig:model_performance_magi}
\end{figure}

\subsection{$Lac$ Operon System}
\label{sec:lac_operon_system}
Among previous methods for DDE inference, performance was typically assessed using simple DDE models with only one or two components with four or fewer parameters to be estimated \citep{wang2012estimating,dondelinger2013ode,wang2022adaptive}, e.g., Hutchinson's equation. In this section, we showcase MAGIDDE's capability to perform inference for a more complex DDE model.

The $lac$ operon is a well-studied gene expression system that enables bacteria to use lactose as an energy source when glucose is scarce \citep{pardee1959genetic}. \citet{yildirim2003feedback} modelled the regulation of induction in the $lac$ operon in $Escherichia$ $coli$
via a system of DDEs to describe the dynamics among five components: 
$\beta$-galactosidase ($B$) converts intracellular lactose ($L$) into allolactose ($A$), which allows DNA transcription and mRNA ($M$) translation to increase the levels of $\beta$-galactosidase ($B$) and permease ($P$). 
The processes of transcription and translation are not instantaneous but require a period of time. Hence, \citet{yildirim2003feedback} incorporated these time delays into the system, to help ensure a satisfactory match between experimental data and the model trajectories. In the following system of DDEs, the time-delay parameter $\tau_M$ represents the time required to produce mRNA from DNA by transcription, $\tau_P$ denotes the translation time between mRNA and permease, and $\tau_B$ is the translation time between mRNA and $\beta$-galactosidase:
\begin{equation*}
    \begin{dcases}
M'(t) &= \alpha_M \left( \frac{1 + K_1 \left(e^{-\mu \tau_M} A(t-\tau_M) \right)^n}{K + K_1 \left(e^{-\mu \tau_M} A(t-\tau_M)  \right)^n} \right) + \Gamma_0 - (\gamma_M + \mu) M(t) \\
B'(t) &= \alpha_B e^{-\mu \tau_B} M(t-\tau_B) - (\gamma_B + \mu) B(t) \\
A'(t) &= \alpha_A B(t)\frac{ L(t)}{K_L + L(t)} - \beta_A B(t)\frac{ A(t)}{K_A + A(t)} - (\gamma_A + \mu) A(t) \\
L'(t) &= \alpha_L P(t)\frac{L_e}{K_{L_e} + L_e} - \beta_{L_1} P(t)\frac{ L(t)}{K_{L_1} + L(t)} - \beta_{L_2}B(t) \frac{ L(t)}{K_L + L(t)} - (\gamma_L + \mu) L(t) \\
P'(t) &= \alpha_P e^{-\mu (\tau_B + \tau_P)} M\left(t-(\tau_B + \tau_P)\right) - (\gamma_P + \mu) P(t)
    \end{dcases}.
\end{equation*}

\citet{yildirim2003feedback} derived formulas for $\alpha_M, \alpha_B,\alpha_P$ according to the steady-state behaviour of the system to calibrate their values; $\mu$ and $\gamma_A$ were calibrated by fitting the model to the experimental data for the induction kinetics of $\beta$-galactosidase; other parameter values were obtained directly from existing literature. Supplementary Table 2 in Appendix H summarizes all of the parameters and their corresponding values from \citet{yildirim2003feedback}. Thus, our goal is to estimate $\tau_B,\tau_M,\tau_P,\gamma_A, \alpha_M,\alpha_B,\alpha_P, \mu,M(0), B(0), A(0), L(0),P(0)$, which include the time-delay parameters, initial conditions, and model parameters calibrated by \citet{yildirim2003feedback}; others are treated as known.

Some experiments recorded measurements of $\beta$-galactosidase every one minute over a 25-minute observation period \citep{alpers1966sequential}. Capturing the transcription delay is still a challenge, as the synthesis of $\beta$-galactosidase and initiations of transcription may occur synchronously in 10 seconds \citep{schwartz1970inactivation}. Other researchers took measurements more frequently at the beginning of the experiments, and gradually less frequently afterwards \citep{alpers1966sequential}. Inspired by these real experiments, we simulate observations from the system with varying denseness from $t=0$ to $t=25$. Specifically, the collection of 23 observation time points for each component is chosen as $\boldsymbol{\gamma} = \{0, 0.25, 0.5, \cdots, 2, 3,4, \cdots,10, 12, 14, \cdots ,20, 25\}$. To mimic the noisy measurements taken from the system, we add Gaussian noise with standard deviation $\sigma_M = 3\times 10^{-5},\sigma_B= 1\times 10^{-5},\sigma_A = 0.02, \sigma_L= 0.01, \sigma_P= 5\times 10^{-4}$ to the model trajectories of each component, which we assume to be known from repeated measures.
 
We generate 100 simulated datasets based on the above setup for the $lac$ operon system. Noting that the observation time points are not equally spaced, the smallest evenly spaced set containing all observation time points is $\boldsymbol{I}_0=\{0,0.25, \cdots, 24.75, 25\}$.
We set $\boldsymbol{I} =\boldsymbol{I}_0$, which provides stable inference without need for a further increase in discretization points. We place uniform priors over $(0,\infty)$ for all model parameters and set $\nu = 2.01$ in the Matern kernel to accommodate rougher system trajectories. 
We run 50,000 HMC iterations (each with 200 leapfrog steps) with the first 25,000 discarded as burn-in. The computational time required for MAGIDDE on each simulated dataset is $\sim$30 minutes on a single CPU core.

Table \ref{tab:lac_est_param} summarizes the average parameter estimates and standard deviations across the 100 simulated datasets. MAGIDDE provides reasonable estimates of  $\tau_B,\tau_P,\gamma_A, \alpha_M,\alpha_B,\alpha_P$, in terms of relatively low standard deviation and bias. However, $\tau_M$ tends to be overestimated, while $\mu$ tends to be underestimated, and both have a high standard deviation. A possible reason is that the observation set is not dense enough to capture the short time-delay of 0.1. Moreover, $\mu$ is closely related to $\tau_M$ as the product $\mu \cdot  \tau_M$ appears in the exponential 
in the first equation of the system. Indeed, the product of the estimated $\mu$ and $\tau_M$ is 0.0032, which is relatively close to the true value of 0.0043.

\begin{table}[!htbp]
\caption{Average parameter estimates obtained by MAGIDDE (with standard deviation after $\pm$ sign) for the $lac$ operon model across 100 simulated datasets.} 
\begin{center}
\begin{tabular}{@{}llllll@{}}
\toprule
\multicolumn{1}{c}{Parameter} & \multicolumn{1}{c}{Truth} & \multicolumn{1}{c}{Estimate} & \multicolumn{1}{c}{Parameter} & \multicolumn{1}{c}{Truth} & \multicolumn{1}{c}{Estimate} \\ \midrule
$\tau_B$                      & 2                         & $2.0024\pm 0.3074$           & $\mu$                         & 0.0226                    & $0.0144\pm 0.0071$           \\
$\tau_M$                      & 0.1                       & $0.2990\pm 0.1275$           & $10^2\times M(0)$                         & 0.0626                   & $0.0630\pm 0.0014$           \\
$\tau_P$                      & 0.83                      & $0.8283\pm 0.3063$           & $10^3\times B(0)$                        & 0                         & $0.0018\pm 0.0014$           \\
$\gamma_A$                    & 0.52                      & $0.4916\pm 0.0255$           & $A(0)$                        & 0.038                     & $0.0319\pm 0.0124$           \\
$10^2\times\alpha_M$                    & $0.0997$                  & $0.0985\pm 0.0032$           & $L(0)$                        & 0.372                     & $0.4018\pm 0.0131$           \\
$\alpha_B$                    & 0.0166                    & $0.0143\pm 0.0011$           & $P(0)$                        & 0.0149                     & $0.0143\pm 0.0004$           \\
$\alpha_P$                    & 10                        & $9.6792\pm 0.3904$           &                               &                           &                              \\ \bottomrule
\end{tabular}
\end{center}
\label{tab:lac_est_param}
\end{table}

Figure \ref{fig:lac_fig} shows the mean inferred trajectories and 95$\%$ pointwise credible intervals obtained by MAGIDDE across the 100 simulated datasets, which indicate that our method can reliably recover the trajectories for the five components. Most of the true model outputs are well-covered by the narrow interval, except for a few data points from the $L$ and $P$ components during the first half minute. The 95$\%$ credible intervals for the $L$ component look different from the other four components, in terms of having a more ``bumpy'' shape; this is due to a small estimated GP bandwidth hyper-parameter ($\phi_2$) to accommodate the sharp increase from the first observation point. The small GP bandwidth in the $L$ component in turn leads to greater uncertainty in its estimated trajectory at discretization time points that are farther away from observations.

\begin{figure}[!htbp]
    \centering
    \includegraphics[width = \textwidth]{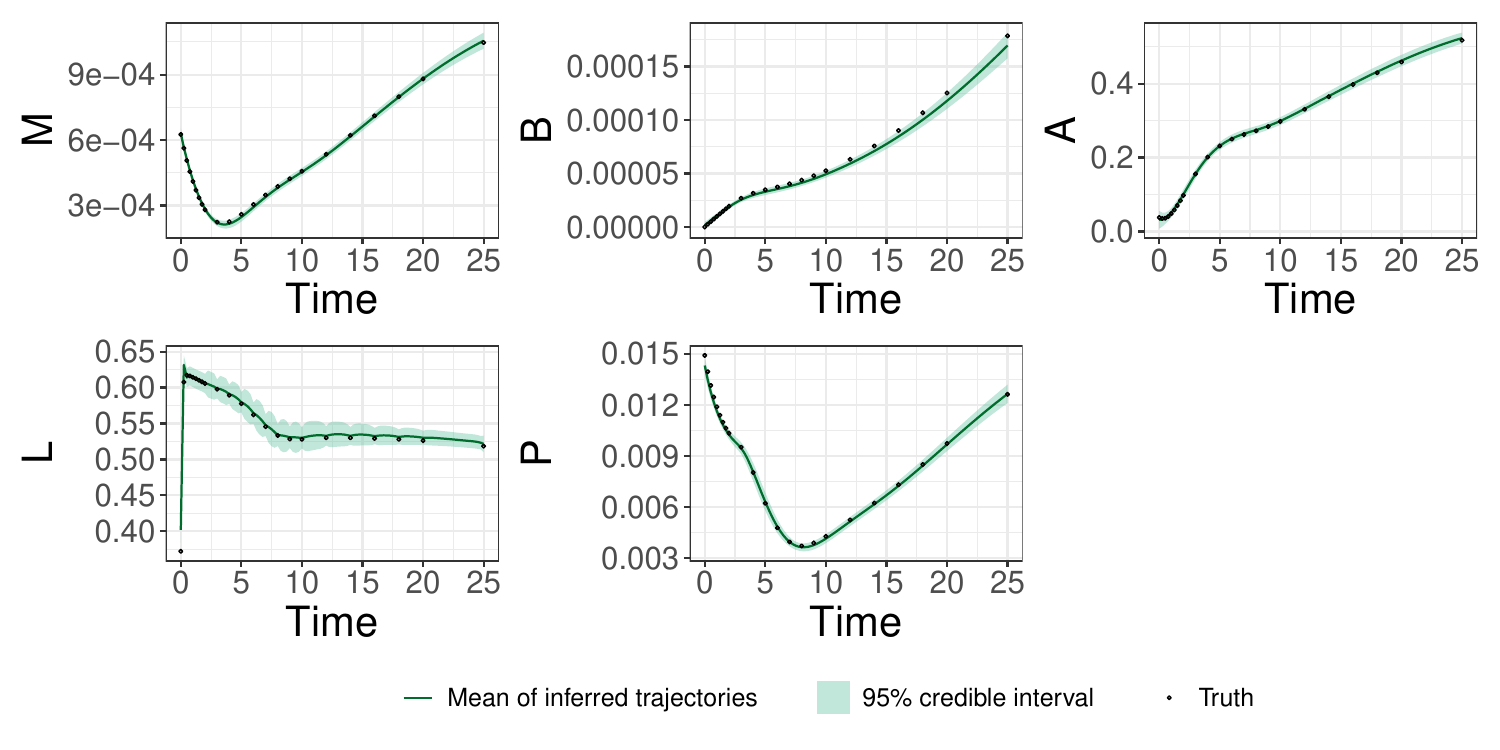}
    \caption{Inferred trajectories obtained by MAGIDDE for each component of the $lac$ operon system over 100 simulated datasets. The green solid line is the mean of the inferred trajectories. The shaded area represents the corresponding 95$\%$ pointwise credible interval, which is constructed by averaging the 0.025 and 0.975 quantiles of the inferred trajectories across the simulated datasets. The black dots indicate the true model trajectory at the observation time points, which are unevenly spaced and relatively sparse for the last 15 minutes.}
    \label{fig:lac_fig}
\end{figure}

\section{Application}
\label{sec:application}

This section applies MAGIDDE to estimate the parameters of a time-delayed compartmental model using COVID-19 data from Ontario, Canada. Omicron was the dominant variant in Ontario as of January 2022, and our focus is to infer the parameters for the Omicron variant over the 30-day observation period from January 24 to February 22, 2022, during the peak of the Omicron wave.
To account for the disease incubation period, a time-delay parameter can be introduced to the basic SIR model \citep{ma2004global}. Inspired by \citet{ma2004global}, we adapt their delayed SIR model to accommodate the Ontario COVID-19 data of interest. First, we add a new compartment $D$ to directly model the death counts associated with COVID-19. Second, we represent the population in each compartment as proportions.
Third, we ignore natural birth and death rates since only a short time period is considered. This leads to our DDE system of interest as follows:
\begin{equation}
\begin{dcases}
    S'(t) &= -\tilde{\beta} S(t)I(t-h) \\
    I'(t) &= \tilde{\beta} S(t)I(t-h) - \mu_d I(t) - \lambda I(t) \\
    R'(t) &= \lambda I(t) \\
    D'(t) &= \mu_d I(t)
\end{dcases},
\end{equation}
where $\tilde{\beta}$ is the normalized disease transmission rate that governs the flow from the susceptible ($S$) to the infected ($I$) compartment, $\mu_d$ represents the death rate of infected individuals, and $h$ is the time-delay accounting for the incubation period of the disease (in days), $\lambda$ is the recovery rate that governs the flow from $I$ to the recovered ($R$) compartment. 
Since $1 = S(t) + I(t) + R(t) + D(t)$, we drop the $S'(t)$  
equation as it suffices to model $(I,R,D)$ and recover $S$ 
using this relationship. In Appendix I, we detail the data processing steps to obtain noisy observations of the daily population size in the compartments.

To set up MAGIDDE for this system, we first consider the discretization set. The system is observed at the equally-spaced time points $\boldsymbol{\gamma} = \{0,1,\cdots, 29\}$ days.  We insert one additional discretization time point, equally-spaced between each pair of adjacent observations, i.e., $\boldsymbol{I} = \{0, 0.5, \cdots, 28.5, 29\}$; a further increase to $|\boldsymbol{I}|$ yielded similar inference. Second, we take the Matern kernel with $\nu = 2.01$ as a default choice. 
Third, we set priors for $\tilde{\beta}$, $\mu_d$, and $\lambda$ to be uniform over $(0,\infty)$. To incorporate the \emph{a priori} belief regarding the estimated mean incubation period of 3.5 days for the Omicron variant 
\citep{manica2022intrinsic}, we place an informative prior of $N(3.5,1)$ on the time-delay parameter $h$. Then, we ran 80,000 HMC iterations (each with 25 leapfrog steps)
with the first 40,000 discarded as burn-in.

Table \ref{tab:para_sird_model} summarizes the parameter estimates with 95$\%$ credible intervals. The model fitting results indicate that the estimated incubation period (3.03 days) is slightly shorter than the prior mean of 3.5 days, and the estimated recovery period ($1/0.0751\approx 13$ days) exceeds the typical recovery period (6.87 days) from the existing literature \citep{wise2022covid}. Noting that we used the number of patients admitted in hospital as a proxy to estimate the current infected population, such patients may have needed more time to recover. \citet{shao2022clinical} pointed out that the median recovery period for patients with moderate or acute symptoms was 13 days, which aligns with our model fitting result.

\begin{table}[!htbp]
\caption{Estimated parameters obtained by MAGIDDE for the time-delayed SIRD model, with 95$\%$ credible intervals, based on Ontario COVID-19 data from January 24 to February 22, 2022.  }
\begin{center}
\begin{tabular}{@{}lccccc@{}}
\toprule
Parameter                       & Estimate          & 95\% CI                         & Parameter          & Estimate              & 95\% CI                         \\ \midrule
$\tilde{\beta} $                        & 0.0254            & (0.0227, 0.0285)                & $R(0)$       &      0.0053         & (0.0050, 0.0056)                \\
$h$                            & 3.0360            & (1.2356, 4.7994)                & $10^5\times D(0)$     &   7.4850  & (7.4471, 7.5170)                \\
$10^3\times \mu_d$                        &  0.3327            & (0.3258, 0.3397)                & $10^4\times \sigma_1$   &   3.4730  & (2.6768, 4.5814)                \\
$\lambda$                      & 0.0751            & (0.0730, 0.0774)                & $10^4\times\sigma_2$   &    3.4756  & (2.6522, 4.6196)                \\
$I(0)$                       & 0.0145            & (0.0142, 0.0148)                & $10^7\times\sigma_3$   &    2.1557  & (0.9873, 3.5252)                \\ \bottomrule
\end{tabular}
\end{center}
\label{tab:para_sird_model}
\end{table}

The trajectory of the $S$ compartment is estimated by subtracting the pointwise estimate of the $I, R, D$ compartments from the total population; the same technique is used to construct its 95$\%$ pointwise credible interval. Figure \ref{fig:sird_model_traj} visualizes the estimated trajectories of all four compartments (in terms of population size) using MAGIDDE, which appear reasonable even though the credible intervals do not cover all the observations. Notably, we observe a weekend effect in the hospitalization data, evidenced by a 2-day flatter ``ledge'' after every five observations in the $I$ and $R$ components. This contributes to the inherent noise within the data, and thus the fitted trajectories appear smoother than the actual observations.

\begin{figure}[!htbp]
    \centering
    \includegraphics[width = \textwidth]{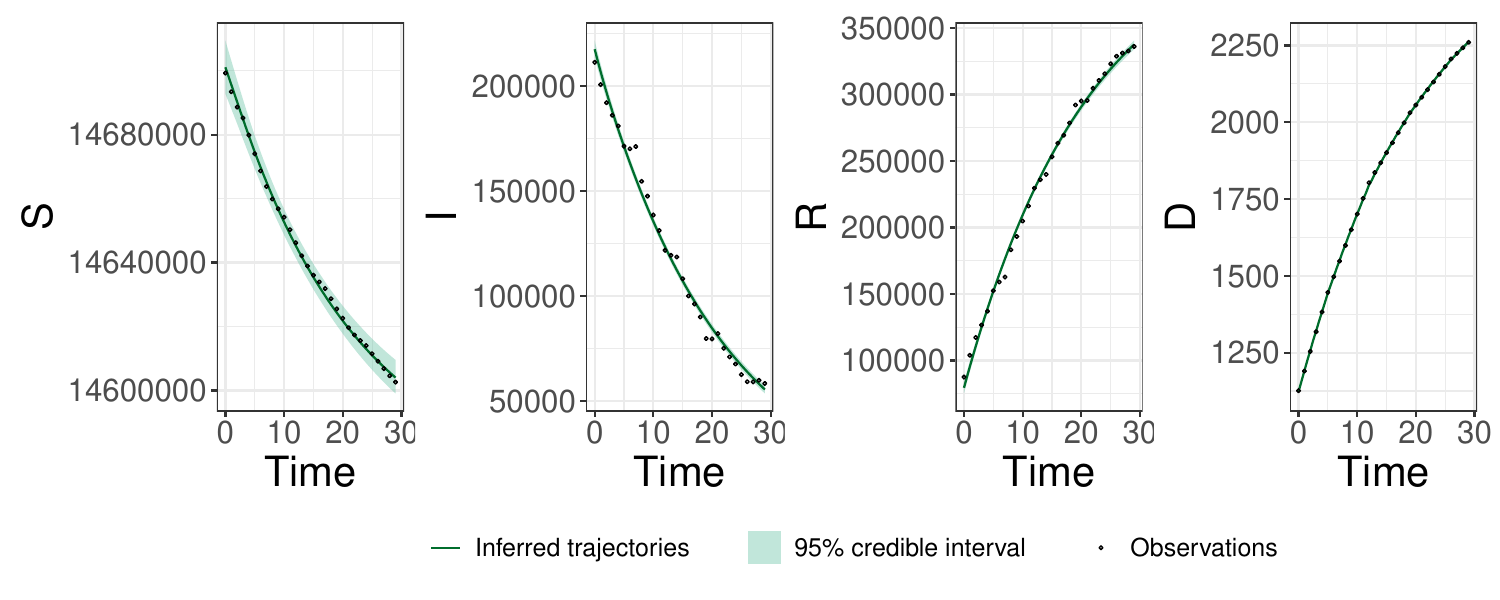}
\caption{Estimated trajectories for four components of the delayed SIR model with the 95 $\%$ credible interval using the Ontario data from January 24 to February 22. The green line represents the inferred trajectory and the green shaded area is the 95$\%$ pointwise credible interval. The black dots are the observed data.}
    \label{fig:sird_model_traj}
\end{figure}

\section{Discussion} \label{sec:discussion}
In this paper, we presented the MAGIDDE method for inference of DDE models. The method provides fast and accurate inference with the help of a linear interpolation scheme for handling historical outputs.
Simulation results indicate that MAGIDDE well-recovers the parameters and trajectory in a simple benchmark system compared to two other representative methods, and a larger \emph{lac} operon system demonstrates the capability of MAGIDDE to recover more complex dynamic systems with multiple time delays. 
Finally, we fit a time-delayed compartmental model to Ontario COVID-19 data as a practical application.

We outline some directions for future work. First, we focused on DDEs with the history function $\mathcal{H}_{\tau_i} = \{x_i(t) = x_i(0), t\in [-\tau_i,0]\}$, primarily due to its prevalence in both theoretical and practical applications \citep{kuang1993delay}.
However, more complex history functions also exist and could be worthy of investigation, as they can significantly alter the dynamics and stability of DDEs.
Second, while DDEs with time-constant parameters cover a broad range of examples, an extension of the method to DDEs with time-varying parameters could also be considered. For example, \citet{liu2020parameter} studied the dynamics of genetic regulatory networks with time-varying delays required for transcription and translation processes. In epidemiology, \citet{pei2023long} introduced the delayed SIR model with a time-varying disease transmission rate and removal rate.

\section*{Acknowledgements}

This work was partially supported by Discovery Grant RGPIN-2019-04771 from the Natural Sciences and Engineering Research Council of Canada.

\bibliography{ref}

\appendix
\renewcommand{\thesection}{Appendix \Alph{section}}
\renewcommand{\thefigure}{Supplementary Figure \arabic{figure}}
\captionsetup[figure]{labelformat=empty} 
\renewcommand{\thetable}{Supplementary Table \arabic{table}}
\captionsetup[table]{labelformat=empty} 

\setcounter{table}{0}
\setcounter{figure}{0}
\setcounter{theorem}{0}

\newpage
\clearpage
\begin{center}
  \LARGE Supplementary Material for Inference for Delay Differential Equations Using Manifold-Constrained Gaussian Processes by Yuxuan Zhao and Samuel W.K. Wong
\end{center}
\title{}
\maketitle
\section{Detailed Expression of Posterior Distribution in MAGI}
\label{sec:magi_detailed}
This section reviews the details of the original MAGI posterior distribution in \cite{yang2021inference}. According to Bayes' rule, the posterior distribution of $\boldsymbol{\theta}$ and $\boldsymbol{x}(\boldsymbol{I})$ given $W_{\boldsymbol{I}} = 0$ and $\boldsymbol{y}(\boldsymbol{\gamma})$ is 
\begin{equation}\begin{aligned}p(\boldsymbol{\theta}, \boldsymbol{x}(\boldsymbol{I})| W_{\boldsymbol{I}} = 0, \boldsymbol{y}(\boldsymbol{\gamma}))\propto p\left(\boldsymbol{\theta}, \boldsymbol{x}(\boldsymbol{I}), W_{\boldsymbol{I}} = 0, \boldsymbol{y}(\boldsymbol{\gamma})\right).
\end{aligned}
\label{eqn:joint_post_bayes}
\end{equation}
Factorizing Equation (\ref{eqn:joint_post_bayes}) yields the four terms,
$$\begin{aligned}
    &p\left(\boldsymbol{\theta},\boldsymbol{x}(\boldsymbol{I}), W_{\boldsymbol{I}} = 0, \boldsymbol{y}(\boldsymbol{\gamma})\right)
    = \underbrace{\pi(\boldsymbol{\theta})}_{(1)} 
\times \underbrace{p\left(\boldsymbol{x}(\boldsymbol{I})|\boldsymbol{\theta}\right)}_{(2)}\times \underbrace{p\left(\boldsymbol{y}(\boldsymbol{\gamma})| \boldsymbol{x}(\boldsymbol{I}),\boldsymbol{\theta}\right)}_{(3)} 
\times \underbrace{p\left(W_{\boldsymbol{I}} = 0| \boldsymbol{y}(\boldsymbol{\gamma}), \boldsymbol{x}(\boldsymbol{I}),\boldsymbol{\theta}\right)}_{(4)}.
\end{aligned}$$

The first term is the prior distribution of the model parameters. Since $\boldsymbol{\theta}$ is independent of $\boldsymbol{x}(\boldsymbol{I})$ from the GP, $p\left(\boldsymbol{x}(\boldsymbol{I})| \boldsymbol{\theta}\right)$ can be simplified to $p\left(\boldsymbol{x}(\boldsymbol{I})\right)$. Similarly, due to the independence between the noisy observations and model parameters, the condition on the parameters can be also be suppressed, yielding $p\left(\boldsymbol{y}(\boldsymbol{\gamma})|\boldsymbol{x}(\boldsymbol{I})\right)$ for the third term. 

We apply the definition of $W_{\boldsymbol{I}} =0$ in the fourth term and note that the distribution of the GP derivative $\boldsymbol{x}'(\boldsymbol{I})$ given $\boldsymbol{x}(\boldsymbol{I})$ is conditionally independent of both the model parameters and the noisy observations; hence the fourth term can be rewritten as $p\left(\boldsymbol{x}'(\boldsymbol{I}) =\mathbf{f} (\boldsymbol{x}(\boldsymbol{I}),\boldsymbol{\theta},\boldsymbol{I})|\boldsymbol{x}(\boldsymbol{I})\right)$, i.e., evaluating the density of $\boldsymbol{x}'(\boldsymbol{I})$ at $\boldsymbol{x}'(\boldsymbol{I}) = \mathbf{f} (\boldsymbol{x}(\boldsymbol{I}),\boldsymbol{\theta},\boldsymbol{I})$.

Combining the above with the fact that the second, third, and fourth terms are all multivariate Gaussian, this leads to the following expression of the posterior: $$
\begin{aligned}
    &\pi(\boldsymbol{\theta})
    \times p\left(\boldsymbol{x}(\boldsymbol{I})\right)\times p\left( \boldsymbol{y}(\boldsymbol{\gamma})|\boldsymbol{x}(\boldsymbol{I})\right) \times p\left(W_{\boldsymbol{I}} = 0| \boldsymbol{x}(\boldsymbol{I})\right)\\
    &\propto \underbrace{ \pi(\boldsymbol{\theta})}_{(1)} \exp \left\{-\frac{1}{2} \sum_{i=1}^m \left[\underbrace{|\boldsymbol{I}| \log (2 \pi)+\log \left|\boldsymbol{C}_i\right|+\left|\boldsymbol{x}_i(\boldsymbol{I})-\boldsymbol{\mu}_i(\boldsymbol{I})\right|_{\boldsymbol{C}_i^{-1}}^2}_{(2)} \right.\right.\\
    &\quad +\underbrace{N_i \log \left(2 \pi \sigma_i^2\right)+\left\|\boldsymbol{x}_i\left(\boldsymbol{\gamma}_i\right)-\boldsymbol{y}_i\left(\boldsymbol{\gamma}_i\right)\right\|_{\sigma_i^{-2}}^2}_{(3)} \\
    &\left.\left.\quad +\underbrace{|\boldsymbol{I}| \log (2 \pi)+\log \left|\boldsymbol{\zeta}_i\right|+\left\|\mathbf{f}_{i, \boldsymbol{I}}^{\boldsymbol{x}, \boldsymbol{\theta}}-\boldsymbol{\mu}'_i(\boldsymbol{I})-\boldsymbol{m}_i\left\{\boldsymbol{x}_i(\boldsymbol{I})-\boldsymbol{\mu}_i(\boldsymbol{I})\right\}\right\|_{\boldsymbol{\zeta}_i^{-1}}^2}_{(4)} \right]\right\},\\
    \end{aligned}
$$ where $\|\boldsymbol{v}\|_A^2=\boldsymbol{v}^{\boldsymbol{\top}} A \boldsymbol{v}$, $|\boldsymbol{I}|$ is the cardinality of $\boldsymbol{I}$, and $\boldsymbol{\mu}_i(\boldsymbol{I})$ is the mean function of the $i$-th component. Moreover, $\sigma_i$ denotes the noise level for component $i$, $N_i$ denotes the number of observations for the $i$-th component, and 
$\mathbf{f}_{i, \boldsymbol{I}}^{\boldsymbol{x}, \boldsymbol{\theta}}$ represents the $i$-th component of $\mathbf{f}\left(\boldsymbol{x}(\boldsymbol{I}), \boldsymbol{\theta}, \boldsymbol{I}\right)$. For each component $i$, the multivariate Gaussian covariance matrices $\boldsymbol{C}_i$ and $\boldsymbol{\zeta}_i$ are computed via
$$
\left\{\begin{array}{ll}
    \boldsymbol{C}_i & =\mathcal{K}_i(\boldsymbol{I}, \boldsymbol{I}) \\
    \boldsymbol{m}_i  & ={ }^{\prime} \mathcal{K}_i(\boldsymbol{I}, \boldsymbol{I}) \mathcal{K}_i(\boldsymbol{I}, \boldsymbol{I})^{-1} \\
    \boldsymbol{\zeta}_i & =\mathcal{K}_i^{\prime \prime}(\boldsymbol{I}, \boldsymbol{I})-{ }^{\prime} \mathcal{K}_i(\boldsymbol{I}, \boldsymbol{I}) \mathcal{K}_i(\boldsymbol{I}, \boldsymbol{I})^{-1} \mathcal{K}^{\prime}_i(\boldsymbol{I}, \boldsymbol{I})
    \end{array},\right.
    $$
    where ${ }^{\prime} \mathcal{K}_i=\frac{\partial}{\partial s} \mathcal{K}_i(s, t), \mathcal{K}^{\prime}_i=\frac{\partial}{\partial t} \mathcal{K}_i(s, t) \text {, and } \mathcal{K}^{\prime \prime}_i=\frac{\partial^2}{\partial s \partial t} \mathcal{K}_i(s, t)$.
    
\section{Conditions of Theorems 1 and 2}
We state the required differentiability assumptions in Conditions \ref{con:1}--\ref{con:3}.
\begin{condition}
   \label{con:1}
 For any $i\in\{1,\cdots,m\}$ and $t\in\boldsymbol{I}$, $f_i(\boldsymbol{x}(t),\boldsymbol{x}(t-\boldsymbol{\tau}),\boldsymbol{\theta},t)$ is continuously differentiable with respect to $\boldsymbol{x}(t-\boldsymbol{\tau})$ in a neighborhood of $\boldsymbol{x}(t-\boldsymbol{\tau})$. 
\end{condition}

\begin{condition}
    \label{con:2}
For any $i,d\in\{1,\cdots, m\}$, any $x_{d,0}$ lying in a neighborhood of $x_d(t-\tau_d)$, and $t\in\boldsymbol{I}$, there exists $C_1>0$, such that $$\bigg|\frac{\partial f_i(\boldsymbol{x}(t),\boldsymbol{x}(t-\boldsymbol{\tau}),\boldsymbol{\theta},t)}{\partial x_d(t-\tau_d)}\big|_{x_d(t-\tau_d) =x_{d,0}}\bigg|<C_1.$$
 \end{condition}
 \begin{condition}
    \label{con:3}
    For any $i, d,k \in\{1,\cdots, m\}$, any $x_{k,0}$ lying in a neighbourhood of $x_k(t-\tau_k)$, and $t\in\boldsymbol{I}$, there exists $C_2>0$, such that $$\bigg|\frac{\partial^2 f_i(\boldsymbol{x}(t),\boldsymbol{x}(t-\boldsymbol{\tau}),\boldsymbol{\theta},t)}{\partial x_d(t-\tau_d)\partial x_k(t-\tau_k)}\big|_{{\substack{x_d(t-\tau_d) =x_{d,0} \\x_k(t-\tau_k) =x_{k,0}}}}\bigg|<C_2.$$ 
 \end{condition}

Condition \ref{con:1} guarantees the differentiability of $f_i(\boldsymbol{x}(t),\boldsymbol{x}(t-\boldsymbol{\tau}),\boldsymbol{\theta},t)$ with respect to $\boldsymbol{x}(t-\boldsymbol{\tau})$. Conditions \ref{con:2} and \ref{con:3} ensure the boundedness of the first and second derivatives in a neighborhood of $\boldsymbol{x}(t-\boldsymbol{\tau})$. Under these three conditions, we have Theorems 1 and 2, which respectively provide the asymptotic error bounds for  $f_i(\boldsymbol{x}(t),\hat{\boldsymbol{x}}(t-\boldsymbol{\tau}),\boldsymbol{\theta},t)$ using the conditional expectation and linear interpolation schemes.

\section{Proof of Theorem 1}
\label{sec:proof_conditional_expectation}
Consider $\Omega = [0,T]$ as the time interval of interest.  Without loss of generality, assume $T = 1$ and suppose the discretization set $\boldsymbol{I} = \{t_1,\cdots, t_n\}$ follows a uniform space-filling design, such that $t_j = \frac{j}{n}$, for $j = 1,\cdots, n$.

Using conditional expectation to approximate the historical output $x_i(t-\tau_i)$ if $t>\tau_i$, we may write $\hat{x}_i(t-\tau_i) = E[x_i(t-\tau_i)|x_i(\boldsymbol{I})]$ for $i = 1,2,\cdots, m$. We denote $\hat{\boldsymbol{x}}(t-\boldsymbol{\tau})=(\hat{x}_1(t-\tau_1),\cdots, \hat{x}_m(t-\tau_m))$. Recall that if $t < \tau_i$, we set $\hat{x}_i(t-\tau_i) = x(0)$ based on the history function.

Define the demeaned Gaussian process $z_i(t) = x_i(t) - \mu_i(t)$, then $z_i(t)\sim \mathcal{GP}(\mathbf{0}, \mathcal{K}_i)$. Then according to Example 1 from \citet{wang2020prediction}, 
$$\sup_{r\in [0,1]}\bigg|E[z_i(r)|z_i(\boldsymbol{I})]-z_i(r)\bigg|=O_P\left(n^{-\nu}(\log(n))^{\frac{1}{2}}\right).$$
Substituting the definition of $z_i(t)$, we obtain
$$\begin{aligned}
    \sup_{r\in [0,1]}\bigg|E[z_i(r)|z_i(\boldsymbol{I})]-z_i(r)\bigg|
    =&\sup_{r\in [0,1]}\bigg|E[x_i(r)-\mu_i(t)|z_i(\boldsymbol{I})]-x_i(r)+\mu_i(t)\bigg|\\
    =&\sup_{r\in [0,1]}\big|\hat{x}_i(r)-x_i(r)\big|\\
    =&O_P\left(n^{-\nu}(\log(n))^{\frac{1}{2}}\right).
\end{aligned}$$
Clearly, when $\nu =2.01$ or $2.5$, $n^{-\nu}\log^{1/2}(n)\to 0$ as $n\to\infty$, so that $$\sup_{r\in [0,1]}\bigg|E[x_i(r)|x_i(\boldsymbol{I})]-x_i(r)\bigg| =o_p(1).$$

Next, consider the uniform deviation 
$$\max_{t\in\boldsymbol{I}}\bigg|f_i(\boldsymbol{x}(t),\boldsymbol{x}(t-\boldsymbol{\tau}),\boldsymbol{\theta},t)- f_i(\boldsymbol{x}(t),\hat{\boldsymbol{x}}(t-\boldsymbol{\tau}),\boldsymbol{\theta},t)\bigg|.
$$
Expanding  $f_i(\boldsymbol{x}(t),\boldsymbol{x}(t-\boldsymbol{\tau}),\boldsymbol{\theta},t)$ around $\boldsymbol{x}(t-\boldsymbol{\tau})=\hat{\boldsymbol{x}}(t-\boldsymbol{\tau})$, 
$$\begin{aligned}&f_i(\boldsymbol{x}(t),\boldsymbol{x}(t-\boldsymbol{\tau}),\boldsymbol{\theta},t)=f_i(\boldsymbol{x}(t),\hat{\boldsymbol{x}}(t-\boldsymbol{\tau}),\boldsymbol{\theta},t)\\
    &\mbox{~~~~~}+\sum_{d=1}^m\frac{\partial f_i\left(\boldsymbol{x}(t),\boldsymbol{x}(t-\boldsymbol{\tau}),\boldsymbol{\theta},t\right)}{\partial x_d(t-\tau_d)}\big|_{x_d(t-\tau_d) =\hat{x}_d(t-\tau_d)} \cdot (x_d(t-\tau_d) -\hat{x}_d(t-\tau_d)) + err
\end{aligned}$$
Since $n^{-\nu}\log^{1/2}(n)\to 0$ as $n\to\infty$, we have $\big|\hat{x}_i(r)-x_i(r)\big| =o_p(1)$.
By the mean value theorem, there exists a $\xi_d$ lying between $x_d(t-\tau_d)$ and $\hat{x}_d(t-\tau_d)$ for $d = 1,\cdots, m$, such that 
$$
\begin{aligned}
    |err| &= \bigg|\frac{1}{2}\sum_{d=1}^m\sum_{k=1}^m\frac{\partial^2 f_i\left(\boldsymbol{x}(t),\boldsymbol{x}(t-\boldsymbol{\tau}),\boldsymbol{\theta},t\right)}{\partial x_d(t-\tau_d)\partial x_k(t-\tau_k)}\big|_{{\substack{ x_d(t-\tau_d) =\xi_d, x_k(t-\tau_k) =\xi_k}}} \\
    &\cdot (x_d(t-\tau_d)-\hat{x}_d(t-\tau_d)) \cdot(x_k(t-\tau_k)-\hat{x}_k(t-\tau_k))\bigg|.\\ 
\end{aligned}
$$
Using Condition 3 yields
$$
\begin{aligned}
    |err| 
    <&\frac{1}{2}m^2\cdot C_2\cdot\sup_{(t-\tau_d)\in [0,1]}\big| (x_d(t-\tau_d)-\hat{x}_d(t-\tau_d)) \big|\cdot \sup_{(t-\tau_k)\in [0,1]}\big|(x_k(t-\tau_k)-\hat{x}_k(t-\tau_k)) \big| \\
    =& O_P\left(n^{-2\nu}\log(n)\right).\end{aligned}
$$
Similarly, for the first-order term, we use Condition 2, obtaining
$$\begin{aligned}
    \bigg|&\sum_{d=1}^m\frac{\partial f_i\left(\boldsymbol{x}(t),\boldsymbol{x}(t-\boldsymbol{\tau}),\boldsymbol{\theta},t\right)}{\partial x_d(t-\tau_d)}\big|_{x_d(t-\tau_d)=\hat{x}_d(t-\tau_d)} \cdot (x_d(t-\tau_d)-\hat{x}_d(t-\tau_d)) \bigg|\\
    &= O_P\left(n^{-\nu}(\log(n))^{\frac{1}{2}}\right).
\end{aligned}$$

Therefore, by the triangle inequality, 
$$\begin{aligned}
&\max_{t\in\boldsymbol{I}}\bigg|f_i(\boldsymbol{x}(t),\boldsymbol{x}(t-\boldsymbol{\tau}),\boldsymbol{\theta},t)- f_i(\boldsymbol{x}(t),\hat{\boldsymbol{x}}(t-\boldsymbol{\tau}),\boldsymbol{\theta},t)\bigg|\\
&=O_P\left(n^{-\nu}(\log(n))^{\frac{1}{2}}\right) + O_P\left(n^{-2\nu}\log(n)\right)\\
&=O_P\left(n^{-\nu}(\log(n))^{\frac{1}{2}}\right).
\end{aligned}$$

\section{Proof of Theorem 2}
\label{sec:proof_linear_interpolation}
Take the settings of the space-filling design to be same as in \ref{sec:proof_conditional_expectation}. We use linear interpolation to approximate the historical output $x_i(t-\tau_i)$ if $t>\tau_i$ for $i = 1,2,\cdots, m$ as described in Section 3.2. Denote $\hat{x}_{i}(t)$ as the linear interpolated values at time points $t$ and $\hat{\boldsymbol{x}}(t-\boldsymbol{\tau})=(\hat{x}_1(t-\tau_1),\cdots, \hat{x}_m(t-\tau_m))$. Let $z_i(t) = x_i(t)-\mu_i(t)$ denote the demeaned Gaussian process. Let $\delta_{i}(t) = x_i(t) -\hat{x}_{i}(t) =  z_i(t) -\hat{z}_{i}(t)$ be the deviation process that quantifies the linear interpolation error, and $\max_{t\in[0,1]} |\delta_{i}(t)|$ denotes the uniform deviation. 
We begin by verifying two lemmas.

 \begin{lemma} 
\label{clm:1}
For any $i\in{1,2,\cdots, m}$, there exists $ L >0,~2\ge \alpha>0, $ such that $E\left[(z_i'(t)-z_i'(s))^2\right]\le L^2|t-s|^{\alpha}$. 
        
        \end{lemma}
        \begin{proof}
           
        Since increasing $\nu$ will increase the smoothness of the Matern covariance function, to have the lemma hold for $\nu = 2.01$ and $\nu = 2.5$, it suffices to verify the lemma for the Matern covariance with any smaller $\nu$.  We will prove this lemma under a rougher covariance function, namely the Matern with $\nu = 3/2$, which is twice-differentiable so that $E\left[(z_i'(t)-z_i'(s))^2\right]$ is well-defined. Note that the Matern covariance function is $2m$ times differentiable if and only if $\nu>m$, $m\in \mathbb{N}^+$ \citep{stein2012interpolation}.

        Suppose $z_i(t)$ is a zero mean Gaussian process, with a Matern covariance of $\nu = 3/2$ that can be written as  $$\mathcal{K}_i (d) = \sigma^2 \left(1+\frac{\sqrt{3}d}{\rho}\right)\exp\left(-\frac{\sqrt{3}d}{\rho}\right),$$
        where $d = |t-s|$, and $0\le d\le 1$.               Consider
        $$\begin{aligned}
            E\left[(z_i'(t)-z_i'(s))^2\right] &= 2\cdot\left[\mathcal{K}'_i(0) -\mathcal{K}'_i(d)\right],\\
        \end{aligned}$$
         where $\mathcal{K}_i'(d)$ is the covariance function for $z_i'(t)$. 
        Since this Matern covariance function is twice-differentiable, the covariance function for $z_i'$ is $\frac{\partial^2 \mathcal{K}_i(s,t)}{\partial s\partial t} =\frac{\partial^2 \mathcal{K}_i(d)}{\partial^2 d}(d)\cdot\frac{\partial d} {\partial s}\cdot\frac{\partial d}{\partial t} = -\frac{\partial^2 \mathcal{K}_i(d)}{\partial d^2}$.
        
        If $t=s$, the inequality $L^2|t-s|^{\alpha} \ge 0 $ clearly always holds.  
        Now considering $t\ne s$, $\mathcal{K}_i'(0) = \frac{\sigma^2}{\rho^2}\cdot\frac{\nu}{\nu-1}|_{\nu = 1.5} =  3\frac{\sigma^2}{\rho^2}$, and
        $$\begin{aligned}
            2\cdot\left[\mathcal{K}'_i(0) -\mathcal{K}'_i(d)\right] =& 2\cdot\left[3\frac{\sigma^2}{\rho^2}+ \frac{\partial^2 K}{\partial d^2}\right]\\
            =&2\cdot\left[3\frac{\sigma^2}{\rho^2} -\frac{3\sigma^2e^{-\frac{\sqrt{3}d}{\rho}}(\rho-\sqrt{3}d)}{\rho^3}\right]\\
            =&6\frac{\sigma^2}{\rho^2}\left[1-e^{-\frac{\sqrt{3}d}{\rho}}\left(1-\frac{\sqrt{3}d}{\rho}\right)\right].
        \end{aligned}$$
        Let $\frac{\sqrt{3}d}{\rho} = q$, then $0<q\le \frac{\sqrt{3}}{\rho}$, and
        \begin{equation}\begin{aligned}
            6\frac{\sigma^2}{\rho^2}\left[1-e^{-\frac{\sqrt{3}d}{\rho}}\left(1-\frac{\sqrt{3}d}{\rho}\right)\right] = 6\frac{\sigma^2}{\rho^2}[1-e^{-q}\left(1-q\right)].
        \end{aligned}
    \label{eqn:substitution}
    \end{equation}
        Now let $w(q) = 1-e^{-q}\left(1-q\right) - 4\cdot q$. $w'(q)= e^{-q}(2-q)-4< e^{0}(2)-4<0$. Then $w(q)$ is monotonically decreasing when $q>0$ and $w(q)<w(0)=0$. Thus, $1-e^{-q}\left(1-q\right)< 4\cdot q$ when $q>0$. Returning to Equation \ref{eqn:substitution}, we have $6\frac{\sigma^2}{\rho^2}\cdot w(q)<0$, or equivalently, $$\begin{aligned}6\frac{\sigma^2}{\rho^2}[1-e^{-q}\left(1-q\right)]< 24\frac{\sigma^2}{\rho^2}\cdot q.
        \end{aligned}$$ 
        Therefore, we can choose $\alpha=1$ and $L = \sqrt{\frac{24\sqrt{3}\sigma^2}{\rho^3}}$ to satisfy $$  E\left[(z_i'(t)-z_i'(s))^2\right]=2\cdot\left[\mathcal{K}'_i(0) -\mathcal{K}'_i(d)\right]\le L^2|t-s|^{\alpha},$$
        which completes the proof. 
        \end{proof}

        \begin{lemma}  
    \label{clm:2}
      When Lemma \ref{clm:1} holds,   $\max_{t \in [0,1]} |\delta_{i}(t)|= O_P\left((\log(n))^{\frac{1}{2}}n^{-1-\alpha/2}\right)$. 
        
        \end{lemma}
        
        \begin{proof}
        
       Let $u = C(2\log(n))^{\frac{1}{2}}$.  According to Theorem 4 from \citet{seleznjev1996large} and Lemma \ref{clm:1}, for any $u>2$, we will have that 
    $$
        \begin{aligned}
            &P\left(\max_{t\in[0,1]} |\delta_{i}(t)|>\epsilon(n)\right)\le  C^{2} K_\alpha C_\alpha^{2} n (2\log(n))^{1/\alpha}\Psi(u)\\
        \end{aligned}
$$
   where $ B_\alpha=2^{-\alpha / 2}, C_\alpha=2^{1+\alpha / 2} $, and  $K_\alpha=2 \exp (5 \cdot 2 / \alpha)(2 \alpha)^{1 / \alpha}$. $\Psi(u) = \phi(u)/u$, where $\phi(u)$ is the standard normal density function.

   By setting $\epsilon(n) = CB_\alpha L(2\log(n))^{\frac{1}{2}}n^{-1+\alpha/2}$ \citep{seleznjev1996large}, where $C$ is an arbitrary constant larger than 1,   \begin{equation}
    \begin{aligned}
        &P\left(\max_{t\in[0,1]} |\delta_{i}(t)|>\epsilon(n)\right)\le  C^{2} K_\alpha C_\alpha^{2} n (2\log(n))^{1/\alpha}\Psi(u)\\
        &\Leftrightarrow P\left(\frac{\max_{t\in [0,1]} |\delta_{i}(t)|}{(\log(n))^{\frac{1}{2}}n^{-1-\alpha/2}}>\sqrt{2}  CB_\alpha L\right)\le C^{2} K_\alpha C_\alpha^{2} n (2\log(n))^{1/\alpha}\Psi(u),\\
    \end{aligned}
    \label{eqn:lemma_1996}
\end{equation}
   and we further investigate the asymptotic behavior of $C^{2} K_\alpha C_\alpha^{2} n (\log(n))^{1/\alpha}\Psi(u)$, or simply $ n (\log(n))^{1/\alpha}\Psi(u)$ after dropping the constant terms.
   
   Note that $\Psi(u)=\phi(u)/u$ is asymptotically equivalent to $\Phi(u)$, where $\Phi(u)$ is the standard normal distribution function, according to the asymptotic behavior of Mill's ratio of the standard normal distribution \citep{wainwright2019high}.  Thus, we alternatively explore the asymptotic behavior of  $n (\log(n))^{1/\alpha}(1-\Phi(u))$.  Using the approximation of $\Phi(u)$ from \citet{bowling2009logistic}  yields 
   $$
    \begin{aligned}n (\log(n))^{1/\alpha} (1-\Phi(u))&\approx n \log(n) \left(1-\frac{1}{1+\exp(-1.5976u-0.07056u^3)}\right)\\
    &=\frac{n (\log(n))^{1/\alpha}}{1+\exp(1.5976u+0.07056u^3)}.
\end{aligned}   
$$
        Replacing $u = C(2\log(n))^{\frac{1}{2}}$ the equation above becomes 
        $$\frac{n (\log(n))^{1/\alpha}}{1+\exp\left(1.5976 C(2\log(n))^{\frac{1}{2}}+0.07056 C(2\log(n))^{\frac{3}{2}}\right)}$$
        By L'Hospital's rule and some algebra, $$\begin{aligned}&\lim_{n\to\infty}\frac{n (\log(n))^{1/\alpha}}{1+\exp\left(1.5976 C(2\log(n))^{\frac{1}{2}}+0.07056 C(2\log(n))^{\frac{3}{2}}\right)}\\
        &=\frac{6250\sqrt{2}n\left(\alpha\log\left(n\right)+1\right)}{C\alpha{\log\left(n\right)}^{\frac{\alpha-2}{2\alpha}}\exp{\left(\frac{\sqrt{2}C\sqrt{\log\left(n\right)}\left(882\log\left(n\right)+9985\right)}{6250}\right)}\left(2646\log\left(n\right)+9985\right)}\\
        &= \frac{6250\sqrt{2}\left(\alpha\log\left(n\right)+1\right)}{C\alpha{\log\left(n\right)}^{\frac{\alpha-2}{2\alpha}}\exp(\frac{9985}{6250})n^{\sqrt{2}C\frac{882}{6250}\sqrt{\log\left(n\right)}-1}\left(2646\log\left(n\right)+9985\right)}\\
        &=0.
        \end{aligned} $$
      As $n\to \infty$, $n (\log(n))^{1/\alpha} (1-\Phi(u))\to 0$, and $C^{2} K_\alpha C_\alpha^{2} n (2\log(n))^{1/\alpha}\Psi(u)\to 0$.
        Then, for any $\epsilon>0$, there exists an $N_1$, such that for all $n> N_1$, we have $C^{2} K_\alpha C_\alpha^{2} n (2\log(n))^{1/\alpha}\Psi(u)<\epsilon$. Denote $M=\sqrt{2}  CB_\alpha L<\infty$ and $a_n = (\log(n))^{\frac{1}{2}}n^{-1-\alpha/2}$. For all $n>N_1$, $$P\left(\bigg|\frac{\max_{t\in[0,1]} |\delta_{i}(t)|}{a_n}\bigg|>M\right)\le C^{2} K_\alpha C_\alpha^{2} n (2\log(n))^{1/\alpha}\Psi(u)<\epsilon.$$
        Therefore, by definition, $\max_{t\in[0,1]} |\delta_{i}(t)|= O_P\left((\log(n))^{\frac{1}{2}}n^{-1-\alpha/2}\right)$, which completes the proof of this part. 
        
        \end{proof}

     Using Lemma \ref{clm:1} and \ref{clm:2}, we can now show that Theorem 2 holds. 
        \begin{proof}

        Consider the uniform deviation $$\begin{aligned}
            &\max_{t\in \boldsymbol{I}}\bigg|f_i\left(\boldsymbol{x}(t), \boldsymbol{x}(t-\boldsymbol{\tau}),\boldsymbol{\theta}, t\right)-f_i\left(\boldsymbol{x}(t), \hat{\boldsymbol{x}}(t-\boldsymbol{\tau}),\boldsymbol{\theta}, t\right)\bigg|.\\
        \end{aligned}
    $$
        Expanding $f_i\left(\boldsymbol{x}(t), \boldsymbol{x}(t-\boldsymbol{\tau}),\boldsymbol{\theta},t\right)$ around $ \boldsymbol{x}(t-\boldsymbol{\tau}) = \hat{\boldsymbol{x}}(t-\boldsymbol{\tau})$, we obtain
        $$\begin{aligned}f_i\left(\boldsymbol{x}(t),\boldsymbol{x}(t-\boldsymbol{\tau}),\boldsymbol{\theta},t\right)&=f_i\left(\boldsymbol{x}(t),\hat{\boldsymbol{x}}(t-\boldsymbol{\tau}),\boldsymbol{\theta},t\right)\\
            &+\sum_{d=1}^m\frac{\partial f_i\left(\boldsymbol{x}(t),\boldsymbol{x}(t-\boldsymbol{\tau}),\boldsymbol{\theta},t\right)}{\partial x_d(t-\tau)}\big|_{x_d(t-\tau_d) = \hat{x}_{d}(t-\tau_d)}\cdot (x_d(t-\tau_d)-\hat{x}_{d}(t-\tau_d)) + err
        \end{aligned}$$
        Clearly, as $n\to \infty$, $(\log(n))^{\frac{1}{2}}n^{-1-\alpha/2}\to 0$. For any $t\in [0,1]$, $|x_i(t)-\hat{x}_{i}(t)|\le \max_{t\in[0,1]} |\delta_{i}(t)|$, and $|x_i(t)-\hat{x}_{i}(t)|=o_p(1).$
        
        By the mean value theorem, there exists a $\xi_d$ for $d = 1,2,\cdots ,m$, lying between $x_d(t-\tau_d)$ and $\hat{x}_{d}(t-\tau_d)$ such that $$\begin{aligned}
            |err| &= \bigg|\frac{1}{2}\sum_{d=1}^m\sum_{k=1}^m\frac{\partial^2 f_i(\boldsymbol{x}(t),\boldsymbol{x}(t-\boldsymbol{\tau}),\boldsymbol{\theta},t)}{\partial x_d(t-\tau_d)\partial x_k(t-\tau_k)}\big|_{{\substack{ x_d(t-\tau) =\xi_d, \\x_k(t-\tau) =\xi_k}}}\cdot (x_d(t-\tau_d)-\hat{x}_{d}(t-\tau_d))(x_k(t-\tau_k)-\hat{x}_k(t-\tau_k))\bigg|\\
        &\le \bigg|\frac{1}{2}\sum_{d=1}^m\sum_{k=1}^m\frac{\partial^2 f_i(\boldsymbol{x}(t),\boldsymbol{x}(t-\boldsymbol{\tau}),\boldsymbol{\theta},t)}{\partial x_d(t-\tau_d)\partial x_k(t-\tau_k)}\big|_{{\substack{ x_d(t-\tau_d) =\xi_d, \\x_k(t-\tau) =\xi_k}}}\bigg|\cdot \max_{t-\tau_d\in[0,1]} |\delta_{d}(t-\tau_d)|\cdot \max_{t-\tau_k\in[0,1]} |\delta_{k}(t-\tau_k)|\\
        \end{aligned}$$
        After using Condition 3, the error term will be      
        $$\begin{aligned}
            |err|  &< \frac{1}{2}\cdot m^2\cdot C_2 \cdot \max_{t\in[0,1]} |\delta_{d}(t)|\cdot \max_{t\in[0,1]} |\delta_{k}(t)|\\
        &= O_P\left(\log(n)\cdot n^{-2-\alpha}\right).\end{aligned}$$
        Similarly, for the first-order term, we apply Condition 2,
        $$\begin{aligned}
            &\bigg|\sum_{d=1}^m\frac{\partial f_i\left(\boldsymbol{x}(t),\boldsymbol{x}(t-\boldsymbol{\tau}),\boldsymbol{\theta},t\right)}{\partial x_d(t-\tau_d)}\big|_{x_d(t-\tau_d) = \hat{x}_{d}(t-\tau_d)}\cdot (x_d(t-\tau_d)-\hat{x}_{d}(t-\tau_d)) \bigg|\\
            &\le \bigg| \sum_{d=1}^m\frac{\partial f_i\left(\boldsymbol{x}(t),\boldsymbol{x}(t-\boldsymbol{\tau}),\boldsymbol{\theta},t\right)}{\partial x_d(t-\tau_d)}\big|_{x_d(t-\tau) = \hat{x}_{d}(t-\tau_d)} \cdot \max_{t-\tau_d\in [0,1]} |\delta_{d}(t-\tau_d)|\bigg| \\
            &= O_P\left((\log(n))^{\frac{1}{2}}n^{-1-\alpha/2}\right).
        \end{aligned}$$

Therefore, for any $i\in\{1,2,\cdots,m\}$, by the triangle inequality,
$$\begin{aligned}
   &\max_{t\in \boldsymbol{I}} \bigg|f_i(\boldsymbol{x}(t),\boldsymbol{x}(t-\boldsymbol{\tau}),\boldsymbol{\theta}, t)- f_i(\boldsymbol{x}(t),\hat{\boldsymbol{x}}(t-\boldsymbol{\tau}),\boldsymbol{\theta}, t)\bigg|\\
    &=O_P\left((\log(n))^{\frac{1}{2}}n^{-1-\alpha/2}\right).
\end{aligned}$$
Finally, since we have shown that taking $\alpha = 1$ satisfies Lemma \ref{clm:1}, the proof is complete.
\end{proof}

\section{Numerical Validation of Approximation Schemes}
\label{sec:Numerical Validation}

The theoretical results of Section 3.2 indicate that MAGIDDE with approximation schemes for the historical outputs in Equations (9) and (10) is reasonable provided that the discretization set $\boldsymbol{I}$ is sufficiently dense. This section provides numerical validation, by examining the effect of $|\boldsymbol{I}|$ on the parameter estimates, inferred trajectories, and computation time under the three schemes for historical outputs described in Section 3: (i) fully Bayesian framework, (ii) conditional expectation, (iii) linear interpolation. 

We generate 100 simulated datasets based on the setup for the log-transformed Hutchinson's equation, given by 
\begin{equation*}
dN(t)/dt = r\left[1-\exp(N(t-\tau))/(1000\cdot K)\right].
\end{equation*}
More details about this benchmark system are provided in Section 4.1. For numerical validation, we  consider the scenario with 16 sparse and equally-spaced observation time points $\boldsymbol{\gamma} = \{0, 2,4, \cdots,30\}$.
We vary the denseness of $|\boldsymbol{I}|$ by adding 0, 1, 3, 7 equally-spaced discretization points between each pair of adjacent observation time points, so that the cardinalities of the corresponding discretization sets are $|\boldsymbol{I}| = \{16, 31, 61, 121\}$. For the purpose of creating a simple experiment to compare schemes and discretization sets, 
in this section only we treat $\sigma^2$ as known, set the starting values for $r, K, \tau$, $\boldsymbol{x}(\boldsymbol{I})$ at the truth, and choose $\nu = 2.5$ for the Matern covariance for faster computation. 
Uniform priors over $[0,5]$ are placed on $r$, $K$, and $\tau$. We set the Matern covariance hyper-parameters $\tilde{\phi}_{1,1}$ and $\tilde{\phi}_{1,2}$ as described in Section 3.3.

For this experiment, the posterior distributions for the fully Bayesian framework, linear interpolation scheme and conditional expectation scheme are all implemented in $Stan$ \citep{carpenter2017stan}, 
which uses NUTS (No-U-Turn sampler, \citet{hoffman2014no}) as the default sampling algorithm. For each scheme, we
run 10,000 MCMC iterations with the first 5,000 discarded as burn-in. 

 \ref{tab:parameter_tab_numerical_validations} summarizes the parameter estimates from the three schemes, across the 100 simulated datasets. We visualize the parameter estimates and RMSEs between the inferred trajectories and the truth for each simulation run using side-by-side boxplots in \ref{fig:parameter_boxplot_numerical_validations} and \ref{fig:rmse_boxplot_numerical_validations}. 
First, considering the effect of $|\boldsymbol{I}|$, the results from all three schemes show a clear decreasing trend in both RMSEs of parameters and inferred trajectories as $|\boldsymbol{I}|$ increases from 16 to 61. At $|\boldsymbol{I}| = 61$, all three schemes provide similarly low RMSEs of the parameters and inferred trajectories. A further increase to $|\boldsymbol{I}| = 121$ yields largely similar parameter estimates as $|\boldsymbol{I}| = 61$ in \ref{fig:parameter_boxplot_numerical_validations}. These results indicate that all the schemes, including fully Bayesian, require $|\boldsymbol{I}| = 61$ as a sufficiently dense discretization set to achieve stable inference for this problem. Second, we consider the effect of the approximation schemes for the historical outputs. When $|\boldsymbol{I}| = 16$, the fully Bayesian framework outperforms the two approximation schemes in recovering the system, and linear interpolation shows significant deviation from the other schemes. When $|\boldsymbol{I}| = 31$, the conditional expectation scheme already performs similarly to fully Bayesian, while linear interpolation still has slightly higher error. This suggests the convergence rate of conditional expectation is faster than linear interpolation, as supported by the results of Theorems 1 and 2. Once $|\boldsymbol{I}|$ reaches 61, the inferences from the three schemes become substantively identical.
This simulation result corroborates the theory, that an increase in the number of discretization points will lead to a more accurate approximation of $f_i(\boldsymbol{x}(t),\hat{\boldsymbol{x}}(t-\boldsymbol{\tau}),\boldsymbol{\theta},t)$, and in turn the posterior distributions of $r, K, \tau$, $\boldsymbol{x}(\boldsymbol{I})$ also become largely indistinguishable between the fully Bayesian and approximate schemes.  Third, computation times follow an expected pattern: linear interpolation scales 
most favorably with $|\boldsymbol{I}|$, while the other schemes become increasingly slow for larger $|\boldsymbol{I}|$. 
While approximation via conditional expectation is about one to two orders of magnitude faster than the fully Bayesian framework, both are hindered by the dense matrix operations required. For $|\boldsymbol{I}| = 61$, which would be recommended for this problem to achieve stable inference (regardless of the scheme used), linear interpolation is about 200 times faster than fully Bayesian and 4 times faster than conditional expectation. This suggests the significantly faster computation speed of linear interpolation is a worthwhile trade-off in practice, even though conditional expectation has a faster theoretical convergence rate.

Therefore, we recommend MAGIDDE with the linear interpolation scheme for practical implementation of the method. To check whether a chosen $\boldsymbol{I}$ is sufficiently dense in practice, we may insert one equally-spaced point between each adjacent pair of existing points in $\boldsymbol{I}$, and re-run MAGIDDE to ensure that the estimates obtained are stable.

\begin{table}[!htbp]
\caption{Supplementary Table 1: Average parameter estimates (with RMSEs over 100 simulated datasets in parentheses) and average running time (in minutes, on a single CPU core) for the log-transformed Hutchinson's equation, comparing the three schemes for historical outputs. Bold highlights the best scheme with respect to the lowest values of parameter RMSEs and shortest running time.}

\resizebox{\textwidth}{!}{
\centering
\begin{tabular}{@{}rlrrrrr@{}}
\toprule
\multicolumn{1}{c}{$|\boldsymbol{I}|$}       & Scheme                  & \multicolumn{1}{c}{$r$} & \multicolumn{1}{c}{$K$} & \multicolumn{1}{c}{$\tau$} & \multicolumn{1}{c}{$N(0)$} & \multicolumn{1}{c}{Runtime}  \\ \midrule
\multirow{1}{*}{16}  & Fully Bayesian          & 0.769($\mathbf{0.033}$)              & 1.990($\mathbf{0.068}$)              & 2.976($\mathbf{0.044}$)                 & 8.128($\mathbf{0.085}$)   &       0.61                             \\
                           & Conditional Expectation & 0.770(0.033)              & 2.046($0.082$)              & 2.878(0.132)                 & 8.125(0.087)   &    0.06                              \\
                           & Linear Interpolation    & 0.850(0.051)              & 1.741(0.265)              & 3.171(0.174)                 & 8.033(0.144)    &       $ \mathbf{0.03}$                       \\ \\
\multirow{1}{*}{31}  & Fully Bayesian          & 0.795(0.008)              & 1.977(0.067)              & 3.014($\mathbf{0.025}$)                 & 8.117($\mathbf{0.063}$)    &     66.7                              \\
                           & Conditional Expectation & 0.795(0.008)              & 1.977($\mathbf{0.067}$)              & 3.015(0.026)                 & 8.116(0.064)   &    0.29                           \\
                           & Linear Interpolation    & 0.798($\mathbf{0.007}$)              & 1.967(0.071)              & 3.027(0.035)                 & 8.109(0.069)   &       $\mathbf{0.11}$                     \\ \\
\multirow{1}{*}{ 61}  & Fully Bayesian          & 0.799($0.004$)              & 1.982(0.064)              & 3.005($\mathbf{0.013}$)                 & 8.140(0.045)       &    251.9                                  \\
                           & Conditional Expectation & 0.799(0.004)              & 1.982($\mathbf{0.064}$)              & 3.005(0.013)                 & 8.139($\mathbf{0.045}$)     &      4.50                       \\
                           & Linear Interpolation    & 0.800($\mathbf{0.004}$)              & 1.979(0.064)              & 3.007(0.014)                 & 8.137(0.045)   &       $\mathbf{1.14}$                        \\ \\
\multirow{1}{*}{121} & Fully Bayesian     & 0.799(0.004)              & 1.986(0.064)              & 3.003($\mathbf{0.012}$)                 & 8.148(0.036)          &    1295.3           \\
                           & Conditional Expectation & 0.799($0.004$)              & 1.985($0.062$)              & 3.003(0.012)                 & 8.148($0.036$)        &       42.2                   \\
                           & Linear Interpolation    & 0.800($\mathbf{0.004}$)              & 1.985($\mathbf{0.062}$)              & 3.004(0.012)                 & 8.147($\mathbf{0.036}$) &      $\mathbf{8.65}$                                   \\ \bottomrule
\end{tabular}
}

\label{tab:parameter_tab_numerical_validations}
\end{table}

\begin{figure}[!htbp]
    \centering
    \includegraphics[width = \textwidth]{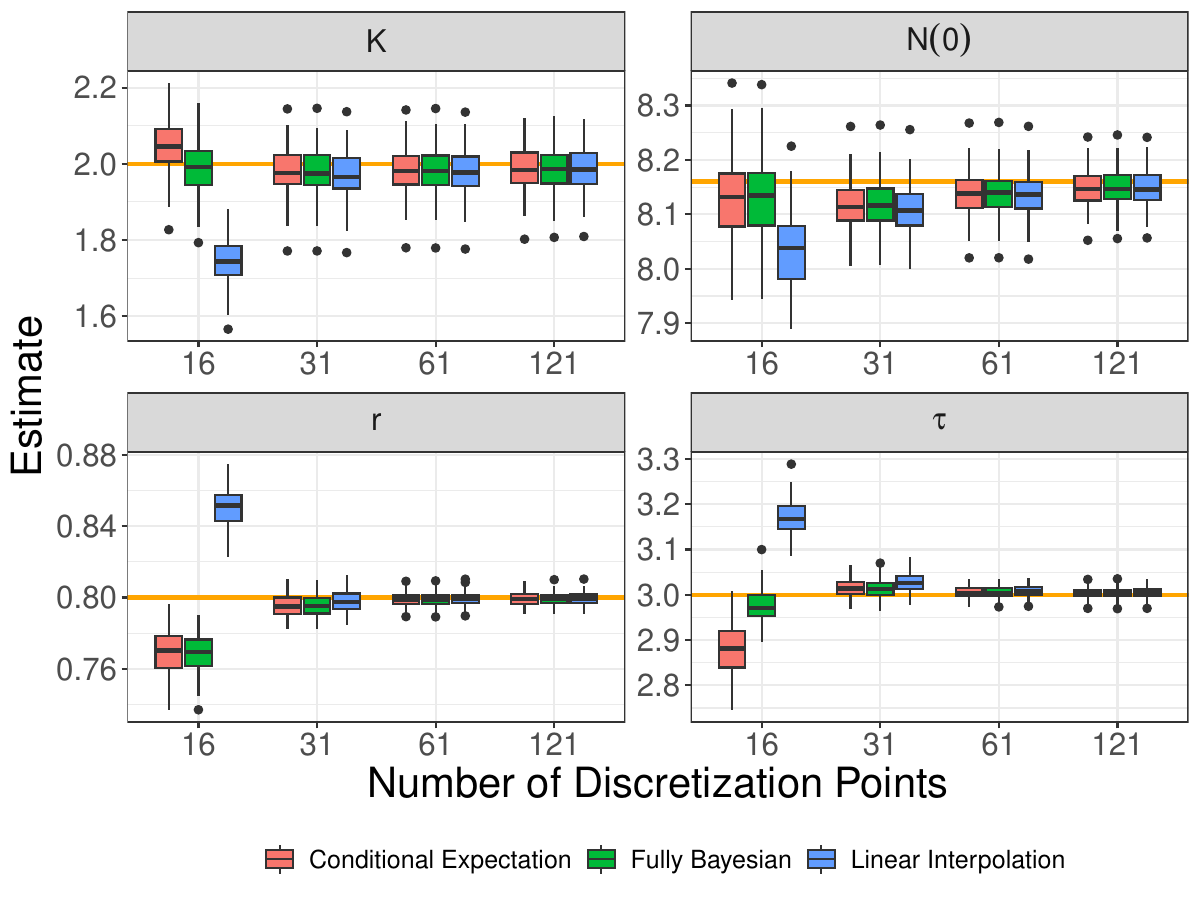}
    \caption{Supplementary Figure 1: Boxplots of parameter estimates for 100 simulated datasets from the log-transformed Hutchinson's equation, comparing the three schemes for historical outputs. The horizontal orange line indicates the true value of the corresponding parameter.}
    \label{fig:parameter_boxplot_numerical_validations}
\end{figure}

\begin{figure}[!htbp]
    \centering
    \includegraphics[width = \textwidth]{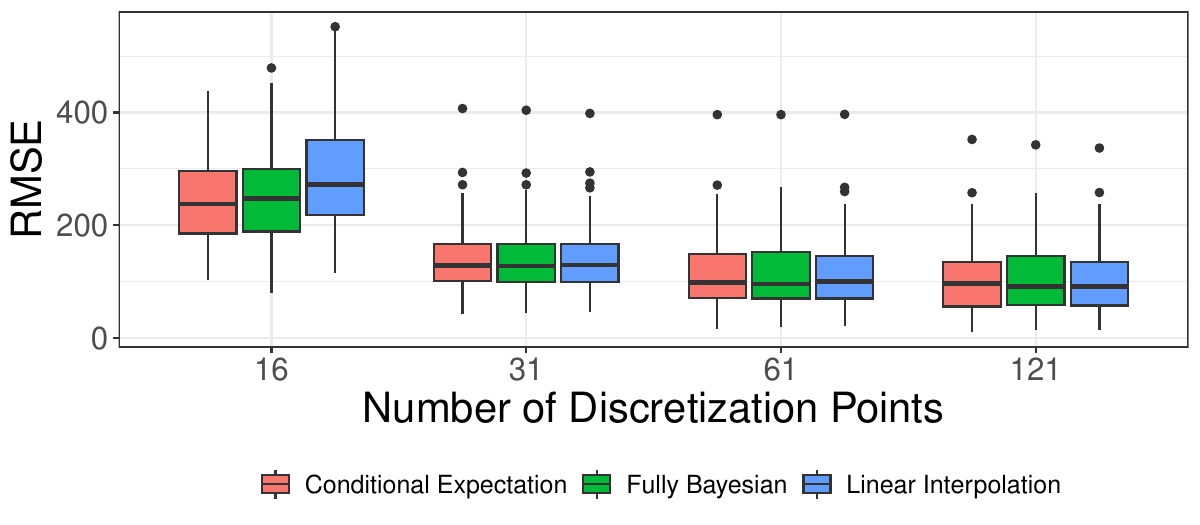}
    \caption{Supplementary Figure 2: Boxplots of RMSEs between the inferred trajectories and the truth, comparing the three schemes for historical outputs over 100 simulated datasets from Hutchinson's equation.}
    \label{fig:rmse_boxplot_numerical_validations}
\end{figure}

\section{Method Comparison Details}
We provide the implementation details of MAGIDDE, deBInfer, and SMCDE as described in Section 4.1. To visualize the performance of each method, we present side-by-side boxplots of parameter estimates and trajectory RMSEs. 
\subsection{Implementation}
For MAGIDDE, we use a Matern covariance with $\nu =2.5$ and an equally-spaced discretization set with $|\boldsymbol{I}| = \max \left(61, |\boldsymbol{\gamma}|\right)$ that is sufficiently dense for all four scenarios, as suggested by Appendix E. 
We obtain the starting values for HMC sampling and estimate the GP hyperparameters as described in Section 3.3. Uniform priors over the positive real numbers are assigned to $K$, $r$, $\tau$, and $\sigma$. 
We run 40,000 HMC iterations (each with 20 leapfrog steps), with the first 20,000 samples discarded as burn-in. 

For deBInfer, we construct a normal likelihood function for noisy observations on the log-scale. 
While we use uniform priors, we restrict their ranges as follows: 
$K$, $r$, and $\tau$ are uniform on the interval $[0,5]$, $N(0)$ is uniform on $[5,10]$, and $\sigma$ is uniform on $[0,0.2]$. Without such restrictions, deBInfer frequently cannot converge to a reasonable result; deBInfer employs 
a variant of the random-walk Metropolis algorithm with a ``reflection'' mechanism to handle the boundaries of the parameter space. 
The proposal variance of each parameter is tuned to ensure that convergence is generally achieved in the simulation runs, and we run 40,000 MCMC iterations with the first 20,000 samples discarded as burn-in.

For SMCDE, we follow the authors' recommendations when running the code. Specifically, one knot is placed at every two time units, the normalized conditional effective sample size threshold is set as 0.9, and resampling threshold is set as 0.5. We use the priors suggested by the authors when they resemble flat priors: $K$ and $r$ are 
$N(0, 5^2)$ restricted to be positive, and $\tau$ is uniform on $[0,50]$. 
For fair comparison, we adjusted the authors' default inverse-gamma$(1,1)$ prior for $\sigma^2$ to be more uniform, with shape parameter 0.001 and scale parameter 0.001. We found that this flatter prior actually led to more accurate parameter estimates, compared to the authors' default. Note that SMCDE expects conjugate priors for some parameters, and these cannot be easily adjusted to be strictly uniform.
\subsection{Supplementary Boxplots for Method Comparison}

 \ref{fig:model_comparison_parameter} and \ref{fig:model_traj_rmse} respectively display the side-by-side boxplots of parameter estimates and trajectory RMSEs for 300 simulated datasets from Hutchison's equation, as obtained from MAGIDDE, deBInfer, and SMCDE. Among these methods, MAGIDDE has the most favorable performance in estimating the parameters and reconstructing the trajectory. The deBInfer method is hindered by a number of outliers where its MCMC sampler did not converge to the correct parameter values, despite tuning the proposal distribution; its trajectory RMSEs also have large outliers as a result. SMCDE exhibits some amount of bias in recovering the system, notably in its estimates for $r$ and $\tau$ when the observations are sparse.

\begin{figure}[!htbp]

\caption{Supplementary Figure 3: Boxplots of parameter estimates for 300 simulated datasets from the log-transformed Hutchinson's equation, comparing the MAGIDDE, deBInfer, and SMCDE methods. The horizontal orange line indicates the true value of the corresponding parameter. The horizontal orange line represents the corresponding true values.}
\centering
\includegraphics[width = \textwidth]{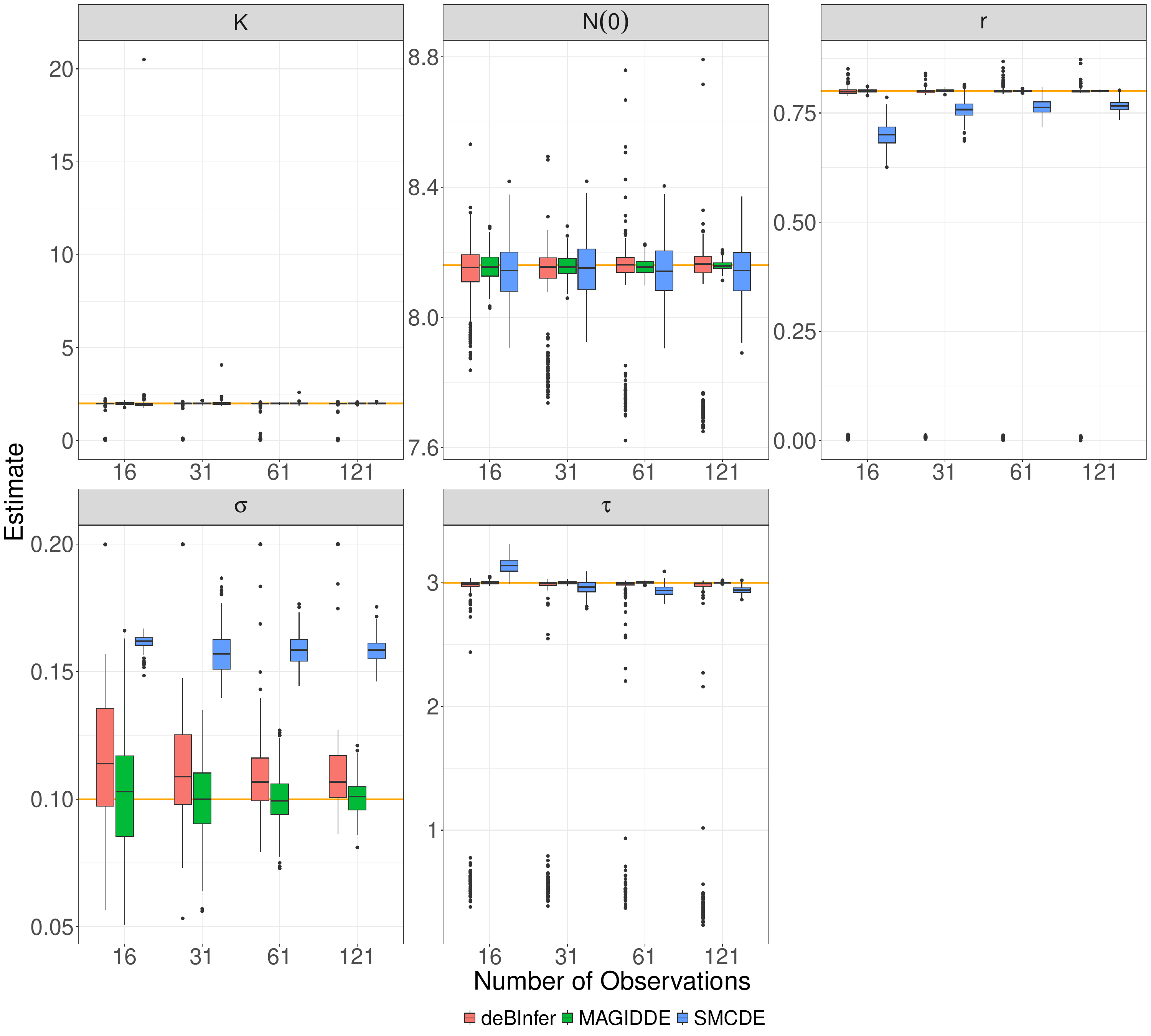}
\label{fig:model_comparison_parameter}
\end{figure}
\begin{figure}[!htbp]
\caption{Supplementary Figure 4: Boxplots of trajectory RMSEs for 300 simulated data sets from Hutchinson's equation, comparing the MAGIDDE, deBInfer, and SMCDE methods.}
\includegraphics[width = \textwidth]{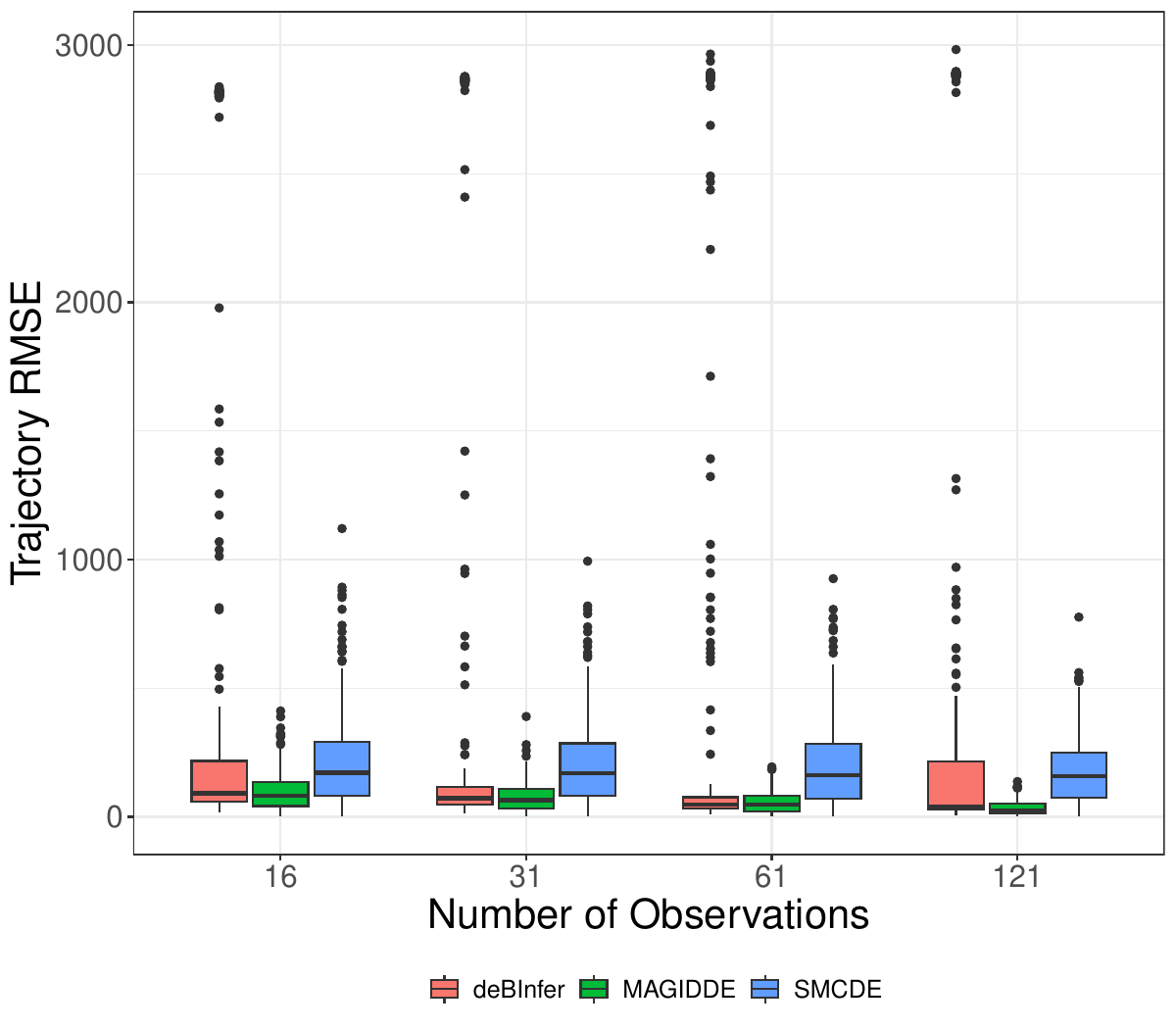}
\label{fig:model_traj_rmse}
\end{figure}

\section{Inferred Trajectory for Hutchison's Equation from 61 Observations using MAGIDDE}

Analogous to Figure 3 in the main text, \ref{fig:infer_traj_61} shows that MAGIDDE well-recovers the true underlying trajectory of Hutchison's equation in the 61 observation scenario; it can be seen that the 95$\%$ pointwise credible interval is narrower with these denser observations, compared to the scenario with 16 observations in Figure 3 of the main text.
\begin{figure}[!htbp]
\centering\includegraphics[width = \textwidth]{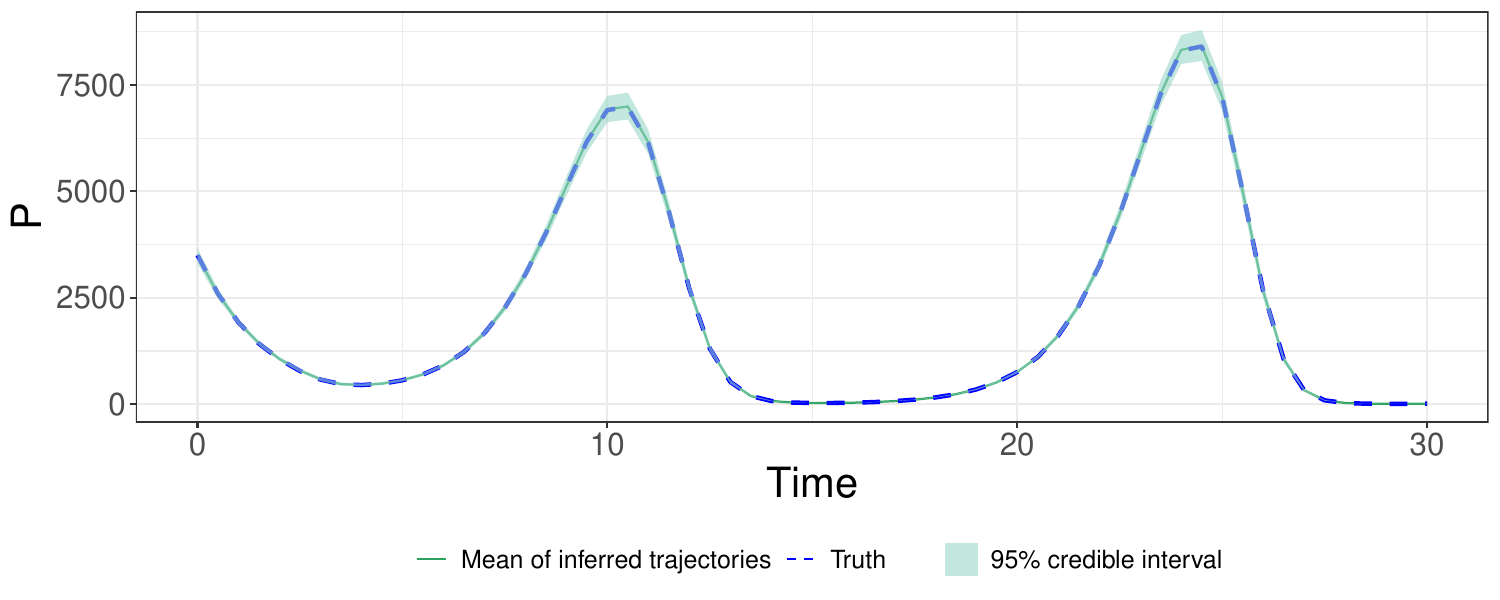}
    \caption{Supplementary Figure 5: Inferred trajectory of Hutchinson's equation from 61 observations using MAGIDDE. The green solid line represents the mean inferred trajectory over 300 simulated datasets, and the blue dashed line is the truth. The green shaded area is the 95$\%$ pointwise credible interval, constructed by taking the average 0.025 and 0.975 quantiles of the inferred trajectories across the simulated datasets. }
    \label{fig:infer_traj_61}
\end{figure}

\section{Summary of Parameters in $Lac$ Operon Model}

\ref{tab:lac_param} summarizes all of the parameters and values in the $lac$ operon model. We estimate the parameters in the first column and others are treated as known from \citet{yildirim2003feedback}.  

\begin{table}[!htbp]
\caption{Supplementary Table 2: Parameters in the $lac$ operon model.}
\begin{center}
\begin{tabular}{@{}crcrcr@{}}
\toprule
Parameter  & \multicolumn{1}{c}{Value} & Parameter  & Value                  & Parameter     & Value                 \\ \midrule
$\tau_B$   & 2                         & $K_1$      & $2.52 \times 10 ^{4}$  & $\beta_{L_1}$ & $2.65 \times 10 ^{3}$ \\
$\tau_M$   & 0.1                       & $K$        & 7200                   & $K_{L_1}$     & 1.81                  \\
$\tau_P$   & 0.83                      & $n$        & 2                      & $\beta_{L_2}$ & $1.76 \times 10 ^{4}$ \\
$\gamma_A$ & 0.52                      & $\Gamma_0$ & $7.25 \times10 ^{-7}$  & $\gamma_L$    & 0                     \\
$\alpha_M$ & $9.97\times 10^{-4}$      & $\gamma_M$ & 0.411                  & $\gamma_P$    & 0.65                  \\
$\alpha_B$ & 0.0166                    & $\gamma_B$ & $8.33 \times 10 ^{-4}$ &               &                       \\
$\alpha_P$ & 10                        & $\alpha_A$ & $1.76 \times 10 ^{4}$  &               &                       \\
$\mu$      & 0.0226                    & $K_L$      & 0.97                   &               &                       \\
$M(0)$     & $6.26\times10^{-4}$       & $\beta_A$  & $2.15 \times 10 ^{4}$  &               &                       \\
$B(0)$     & 0                         & $K_A$      & 1.95                   &               &                       \\
$A(0)$     & 0.038                     & $\alpha_L$ & 2880                   &               &                       \\
$L(0)$     & 0.372                     & $L_e$      & 0.08                   &               &                       \\
$P(0)$     & 0.0149                    & $K_{L_e}$  & 0.26                   &               &                       \\ \bottomrule
\end{tabular}%
\end{center}%
\label{tab:lac_param}
\end{table}

\section{Data Processing for Time-delayed SIRD Model}
In this section, we describe the procedures for processing the Ontario COVID-19 data to obtain noisy observations for the time-delayed SIRD model. We also specifically detail the steps for obtaining the initial observation of each compartment.   
\subsection{Data Processing Steps}
We outline the data processing steps for obtaining noisy observations of the population size in the 
compartments on each day. 
These are then converted into proportions by dividing by the Ontario population of 14,999,441 in the first quarter of 2022 \citep{canadapopulation}. Following \citet{zhao2023comparative}, we approximate the daily population size in the $I$ compartment by dividing the number of hospitalized COVID-19 cases \citep{patients_in_hospital} by the overall COVID-19 hospitalization rate. 
We convert the weekly death counts \citep{death} into daily death counts by linear interpolation; then, the population in the $D$ compartment on a given day is approximated by the cumulative sum of daily death counts up to that day. The population in the $S$ compartment on a given day is estimated by subtracting the daily confirmed cases \citep{daily_confirmed_cases} from the previous day's $S$ compartment population. Given these $S,I,D$ population estimates for each day, $R$ then consists of the remainder of the population. 
The data setup further requires the population of each compartment on the initial day of the observation period; this step is detailed in the next subsection.

\subsection{Initial Noisy Observations for the Time-Delayed SIRD Model}

The first wave of COVID-19 driven by the Omicron variant began in Ontario, Canada on January 7, 2022. This wave was marked by a record-high number of hospital admissions \citep{cbc_news}. Due to the high reinfection rate of the Omicron variant, previously infected individuals remained susceptible; hence we assume the flow from the $S$ compartment to the $I$ compartment starts from January 6 and that there are no deaths related to the Omicron variant prior to January 7. We use the following procedure to approximate the population in the $S,I,R,D$ compartments on the initial day (January 24) of the observation period under these two assumptions. 

A noisy observation for the $I$ compartment on January 24 can be obtained by dividing the number of patients hospitalized on that day by the hospitalization rate of the Omicron variant \citep{zhao2023comparative}. The trend of $I$ observations on the few days before January 24 is relatively flat, which aligns well with the history function assumed for the model.
Furthermore, based on our assumptions, a noisy observation for the $D$ compartment on January 24 is obtained by the cumulative sum of the daily death count from January 7 through January 24.

To obtain noisy observations of $S$ and $R$ on January 24, we need estimates for all four compartment sizes at the beginning of the Omicron wave (January 6). 
First, we obtain a noisy observation for $I$ on January 6 using the aforementioned hospitalization rate method. Second, based on our two assumptions we set $R$ and $D$ to be zero on January 6, and so $S$ on January 6 is estimated by subtracting $I$ from the total population. 
Third, a noisy observation of $S$ on January 24 is obtained by taking the estimated $S$ on January 6 and subtracting the cumulative case counts from January 6 to January 24. Last, we subtract these estimates of $S$, $I$, and $D$ on January 24 from the total population to obtain a noisy observation of $R$ on January 24.

\end{document}